%% file: main.tex
\documentclass[10pt,conference]{IEEEtran}
\usepackage[cmex10]{amsmath}
\usepackage{amsfonts}
\usepackage{array}

\input{commands}

\usepackage{wrapfig}
\usepackage{graphicx}
\begin{document}

%\title{SAT-solving And Model Checking By Computing Range Reduction}
\title{Bug Hunting By Computing Range Reduction}

\author{\IEEEauthorblockN{Eugene Goldberg and Panagiotis Manolios} 
\IEEEauthorblockA{
College of Computer and Information Science\\
Northeastern University, USA, 
360 Huntington Ave., Boston MA 02115, USA \\
email:
 \{eigold,pete\}@ccs.neu.edu}}

\maketitle

\input{abstract}
\input{i1ntroduction}

\input{e2xample}

\input{p1qe}

\input{d1efinitions}
\input{m0ain_idea_mc_ccr}
\input{isol_publ_cex}

\input{bug_hunting}

\input{p2roving_corr}

\input{c3omparison}

\input{m1c_crr}

\input{experiments}

\section*{Acknowledgment}
This research was supported in part by DARPA under AFRL
Cooperative Agreement No.~FA8750-10-2-0233 and by NSF grants
CCF-1117184 and CCF-1319580.
\input{c5onclusions}

\input{p3roofs}
\bibliographystyle{plain}
\bibliography{short_sat,local}
\end{document}

%% file: commands.tex
\newtheorem{definition}{Definition} 
\newtheorem{proposition}{Proposition}
\newtheorem{lemma}{Lemma}

\newcommand{\pnt}[1]{{\mbox{\boldmath $#1$}}}
\newcommand{\init}{{\mbox{\boldmath $s_{\mi{init}}$}}}

\newcommand{\V}[1]{\mbox{$\mathit{Vars}(#1)$}}

\newcommand{\s}[1]{\mbox{$\{#1\}$}}

\newcommand{\nGz}[2]{$G_{non-\{z\}}$}

\newcommand{\mi}[1]{\mathit{#1}}
\newcommand{\ti}[1]{\textit{#1}}
\newcommand{\tb}[1]{\textbf{#1}}

\newcommand{\ttt}{\>\>\>}

\newcommand{\Tt}{\>\>}

\newcommand{\prob}[2]{\mbox{$\exists{#1} [#2]$}}

\newcommand{\Comment}[1]{}

\newcommand{\PQE}{\ti{DS-PQE}~}
\newcommand{\pqe}{\ti{DS-PQE}}

\newcommand{\inp}[2]{\mbox{$(\pnt{#1},\pnt{#2})$}}

\newcommand{\MC}{\mbox{$\mi{MC\_CRR}$}}
\newcommand{\mc}{\mbox{$\mi{MC\_CRR}$~}}
\newcommand{\Impl}[2]{\mbox{$#1 \rightarrow #2$}}
\newcommand{\Nmpl}[2]{\mbox{$#1 \not\rightarrow #2$}}
\newcommand{\CDI}{\ti{DCDS}}
\newcommand{\Cdi}{\ti{DCDS}~}
\newcommand{\ecnf}{\ensuremath{\exists\mathrm{CNF}}}

\newcommand{\IP}[2]{(\pnt{#1},\pnt{#2})}
\newcommand{\TR}[3]{\mbox{(\pnt{#1_{#2}},\dots,\pnt{#1_{#3}})}}
\newcommand{\tr}[4]{\mbox{$\IP{#1_{#3}}{#2_{#3}},\dots,\IP{#1_{#4}}{#2_{#4}}$}}
\newcommand{\Tr}[4]{\mbox{$\IP{#1_{#3}}{#2_{#3}},\dots,\IP{#1_{#4}}{#2_{#4}},\dots$}}
\newcommand{\rr}{\mbox{$\mathbb{H}^*$}}
\newcommand{\ur}{\mbox{$\mathbb{U}$}}

\newcommand{\vl}[1]{\mbox{$\mathit{Val}(\pnt{#1})$}}

%% file: abstract.tex
\begin{abstract}
We describe a method of model checking called Computing Range
Reduction (CRR).  The CRR method is based on derivation of clauses
that reduce the set of traces of reachable states in such a way that
at least one counterexample remains (if any). These clauses are
derived by a technique called Partial Quantifier Elimination (PQE).
Given a number $n$, the CRR method finds a counterexample of length
less or equal to $n$ or proves that such a counterexample does not
exist\footnote{\input{f0ootnote.tex}}.  We show experimentally that a
PQE-solver we developed earlier can be efficiently applied to
derivation of constraining clauses for transition relations of
realistic benchmarks.

One of the most appealing features of the CRR method is that it can
potentially find long counterexamples. This is the area where it can
beat model checkers computing reachable states (or their
approximations as in IC3) or SAT-based methods of bounded model
checking.  PQE cannot be efficiently simulated by a SAT-solver. This
is important because the current research in model checking is
dominated by SAT-based algorithms. The CRR method is a reminder
that one should not put all eggs in one basket.
\end{abstract}

%%  LocalWords:  CRR EquiSatisfiability

%% file: f0ootnote.tex
To make exposition simpler, in this paper, we formulate a version of
the CRR method that proves only that a property holds for $n$
transitions. However, the CRR method can be modified to become
complete and hence able to prove that a property holds for an
arbitrary number of transitions. We are planning to publish this
modification of the CRR method in the near future.

%% file: i1ntroduction.tex
\section{Introduction}
In this paper, we introduce a new method of model checking called
Computing Range Reduction (CRR).

\subsection{Motivating example}
Let $\xi$ be a state transition system. Let formula $I$ specify the
initial states of $\xi$ and formula $P$ specify a property that does
not hold for $\xi$.  Suppose that we know that there exists \ti{only
  one} sequence $D=\TR{s}{0}{n}$ of states from an initial state to a
bad state. Suppose that we know only the state \pnt{s_0} of this
sequence (that is an initial state of $\xi$) and want to find the
remaining states \pnt{s_1},\dots,\pnt{s_n}.  We assume here that a
\tb{counterexample} $E$ is a trace \tr{s}{x}{0}{n-1} where \pnt{x_i}
is a complete assignment to combinational input variables in $i$-th
time frame. As usual, we assume that $\xi$ transitions to state
\pnt{s_{i+1}} from \pnt{s_{i}} under assignment \pnt{x_i},
$i=0,\dots,n-1$. So trace $E$ leads to state \pnt{s_n}. A traditional
model checker cannot exploit the fact that every counterexample $E$
goes through the same sequence of states $D$. To find $E$, such a
model checker would have to build a sequence of sets of states
$A_1,\dots,A_n$ where $A_i$ is the set of states \ti{reachable} from
\pnt{s_0} in $i$ transitions or an over-approximation thereof. For the
sake of simplicity, below, we assume that $A_i$ is the precise set of
states reachable from \pnt{s_0} in $i$ transitions.

In reality, finding a counterexample $E$ does not require computing
sets $A_i$,$i=1,\dots,n$.  Let $B_i$ denote the set of states that are
reachable in $i$ transitions from initial states \ti{different} from
\pnt{s_0}. (We assume here that $I$ specifies more than one initial
state.)  Every state \pnt{s_i} of $D$ is in $A_i \setminus B_i$,
$i=1,\dots,n$. Indeed, \pnt{s_i} cannot be in $B_i$ because then $D$
would not be the only  sequence of states leading to a bad
state. Importantly, the size of the set $A_i \setminus B_i$ can be
dramatically smaller than $A_i$.
 
The CRR method is able to find a counterexample $E$ by generating only
sets of states $A_i \setminus B_i$.  Here is how it is done.  Instead
of finding the set of states reachable from state \pnt{s_0} in $i$
transitions, the CRR method builds the set of states that become
\ti{unreachable} in $i$ transitions if the state \pnt{s_0} is
\ti{excluded} from the set of initial states. It is not hard to see
that this is exactly the set $A_i \setminus B_i$ because the latter
consists of states that are reachable in $i$ transitions \ti{only}
from state \pnt{s_0}. Obviously, these states become unreachable if
\pnt{s_0} is excluded.  The fact that set $A_n \setminus B_n$ contains
a \ti{bad} state \pnt{s_n} means that \pnt{s_n} is reachable from
\pnt{s_0} in $n$ transitions. Hence property $P$ fails.

\subsection{Operation of the CRR method in more detail}
\label{subsec:intr_operation}
Let us use the example above to describe the operation of the CRR
method in more detail. Let $N$ be a circuit specifying the transition
relation of system $\xi$. Let $S$ and $X$ be the sets of state
variables and combinational input variables of $N$ respectively. So $S
\cup X$ is the set of input variables of $N$.  The key operation of
the CRR method is to \tb{exclude} some input assignments of the
initial time frame and compute the set of reachable states that become
unreachable due such an exclusion.

In our example, the set of excluded inputs is specified by clause $C$
that is falsified only by state \pnt{s_0}. This clause excludes every
assignment \IP{s_0}{x} where \pnt{x} is an arbitrary complete
assignment to $X$.  To compute the effect of constraining inputs of
the initial time frame by $C$, a set of \tb{range reduction formulas}
$H_1,\dots,H_n$ is constructed. (We assume that time frames are
indexed starting with 0. So the initial time frame has index 0.)
Formula $H_i$ evaluates to 0 for state \pnt{s} iff \pnt{s} is
reachable in $i$-transitions but becomes unreachable in
$i$-transitions after removing the traces excluded by $C$. In our
example, clause $C$ excludes every trace that starts with state
\pnt{s_0}. So, the set of states falsifying formula $H_i$ is equal to
$A_i \setminus B_i$.

The name ``range reduction formula'' is due to the fact that $H_i$
specifies the reduction of the range of a combinational circuit caused
by excluding its inputs by $C$. This circuit is a composition of $i$
copies $N$. Formulas $H_i$ are computed one by one.  Once formula
$H_i$ is formed, the CRR method checks if \Impl{\overline{H_i}}{P}
holds. If it does not, then there is a bad state \pnt{s_i} that
becomes unreachable in $i$ transitions after excluding state
\pnt{s_0}.  Hence \pnt{s_i} is reachable from \pnt{s_0} and $P$ does
not hold.  If \Impl{\overline{H_i}}{P} holds, then the CRR method
computes the next range reduction formula $H_{i+1}$.  This goes on
until a bad state falsifying the most recent range reduction formula
is found.

So far, we assumed that there exists only one sequence of states $D$
from an initial state to a bad state and this sequence specifies
counterexamples of length $n$. Suppose that this is not the case.
That is either property $P$ holds for $n$ transitions or for every $i$
less or equal to $n$, clause $C$ does not exclude all counterexamples
of length $i$ (if any).  Then, \Impl{\overline{H_i}}{P} holds for
every range reduction formula $H_i$. This means that excluding the
inputs of the initial time frame falsified by clause $C$ does not
affect the answer to the question whether $P$ holds for $n$
transitions. In this case, we will say that clause $C$ is
\pnt{P^n}\tb{-equivalent}.

Given a number $n$, the CRR method either finds a counterexample of
length at most $n$ or proves that $P$ holds for $n$ transitions. The
latter is done by generating $P^n$-equivalent clauses until one of the
two conditions below is met.
\begin{enumerate}
\item The set of all possible traces of $n$ transitions reduces to one
  trace consisting of only good states.
\item The set of all possible traces of $m$ transitions where $m \leq n$
reduces to one trace $L$ where
\begin{itemize}
\item all states of $L$ are good and
\item  the last  state of $L$ repeats some previous state of $L$
\end{itemize}
\end{enumerate}

\subsection{What sets CRR method apart from competition}
One of the most appealing features of the CRR method is that it can
potentially detect \tb{very deep bugs}.  Such bugs are hard to find by
the existing methods.  The complete methods based on computing
reachable states or their over-approximation work in a
\ti{breadth-first} manner.  That is they consider counterexamples of
length $n$ only \ti{after} they proved that no counterexample of
length $n-1$ or less exists.  This also applies to Bounded Model
Checking (BMC).  The breadth-first search strategy makes these methods
very inefficient in finding deep bugs.  As we mentioned above, when
the CRR method looks for a counterexample, it generates range
reduction formulas $H_1, \dots H_n$.  This means that the CRR method
looks for counterexamples in a \ti{depth-first} manner.  In
particular, the CRR-method can find a counterexample of length $n$
without proving that counterexamples of length less than $n$ do not
exist. This can be done efficiently because the CRR method computes
only a small subset of the set of states reachable in $i$ transitions
$i=1,\dots,n$.

\subsection{Partial quantifier elimination}

Computing a range reduction formula $H_i$ comes down to solving an
instance of the Partial Quantifier Elimination (PQE)
problem~\cite{tech_rep_pqe,south_korea}. In general, a PQE-solver
cannot be efficiently simulated by a SAT-solver. This is important
because the current research in model checking is dominated by
SAT-based approaches. The CRR method is a reminder that one should not
put all eggs in one basket.

In the experimental part of the paper, we give some results of
applying our PQE-algorithm~\cite{tech_rep_pqe} to constructing range
reduction formulas.  We compute such formulas for transition relations
of the HWMCC-10 benchmarks.  Our experiments show that even the
current version of the PQE algorithm that has huge room for
improvement can be successfully applied to computing range reduction
formulas.

\subsection{Structure of the paper}
This paper is structured as follows. In Section~\ref{sec:example}, we
present a simple example illustrating the operation of the CRR method.
We also discuss the advantages of the CRR method in finding deep
bugs. Section~\ref{sec:pqe} gives a brief introduction into partial
quantifier elimination. Basic definitions are given in
Section~\ref{sec:defs}. In Section~\ref{sec:main_idea}, we explain the
main idea of the CRR-method. Section~\ref{sec:isol_publ_traces}
introduces the important classification of traces as isolated or
public with respect to a constraining clause.  Application of CRR to
bug hunting is discussed in Section~\ref{sec:bug_hunting}.  In
Section~\ref{sec:p_equiv_clauses}, we explain how the CRR method
identifies $P^n$-equivalent clauses.  We compare the CRR method with
other model checkers in Section~\ref{sec:comparison}.
Section~\ref{sec:mc_crr} describes a model checker called \mc that is
based on the CRR method.  Experimental results are given in
Section~\ref{sec:experiments}. In Section~\ref{sec:conclusions}, we
make some conclusions.

%% file: e2xample.tex
\section{An Example Of How CRR Method Operates}
\label{sec:example}
In this section, we describe the operation of the CRR method when
checking a property of an \ti{abstract} $k$-bit counter.  An abstract
counter is a regular counter where no assumptions about the binary
encodings of numbers are made. In particular, a pair of consecutive
numbers can have completely different binary representations.

One can view an abstract counter as describing a sub-behavior of a
sequential circuit going through a long sequence of states $K$ where
all states of $K$ are unique. The counter has a combinational variable
$x$ whose value specifies whether this counter stays in the current
state or moves to the next state of $K$. Since an abstract counter is
meant to simulate a long sequence of unique states of an arbitrary
sequential circuit, it is reasonable to avoid making any assumptions
about the way states are encoded.

Subsection~\ref{subsec:exmp_descr} describes the example with an
abstract counter in more detail.  Application of the CRR method to
this example is described in Subsection~\ref{subsec:appl_crr_method}.
Subsection~\ref{subsec:deep_bugs} uses an abstract counter to show the
advantage of the CRR method over existing methods in finding deep
bugs.

\subsection{Problem description}
\label{subsec:exmp_descr}
An abstract $k$-bit counter is specified by a sequential circuit $\xi$
defined as follows.  Let $S=\s{s_1,\dots,s_k}$ specify the set of
state variables of $\xi$ and $x$ be the only combinational input
variable of $\xi$. We will assume that $\xi$ has only one initial
state where $s_i=0$,$i=1,\dots,k$. We will denote the initial state as
\init. Let \vl{s} denote the number stored by the counter in state
\pnt{s}. As we mentioned above, we do not assume any relation between
\vl{s} and \pnt{s}. Our only constraints are that $\pnt{s} \neq
\pnt{s^*}$ implies $\vl{s} \neq \vl{s^*}$ and that \vl{\init} = 0.

 The transition relation of $\xi$ is specified as follows.  Let
 \pnt{s} be the current state of $\xi$.
\begin{itemize}
\item If $x=0$, $\xi$ remains in state \pnt{s}.
\item If $x=1$,
\begin{itemize}
\item if $\vl{s} \neq 2^k-1$,  $\xi$  switches to state \pnt{s'}
 such that \vl{s'}~:= \vl{s} + 1.
\item if $\vl{s} = 2^k-1$, $\xi$ resets i.e. switches to \init.
\end{itemize}
\end{itemize}

Let $P(S)$ be a formula such that $P(\pnt{s})=1$ iff $\vl{s} < d$.
The problem we want to solve is to check if $\xi$ satisfies property
$P$ for $n$ transitions. To prevent resetting the counter, we will
assume that $n < 2^k-1$. Since \vl{\init}=0 and one transition
increases the value corresponding to the current state by at most 1,
$P$ holds if $n < d$ and fails otherwise.
\subsection{Application of the CRR method}
\label{subsec:appl_crr_method}
Here is how the problem above is solved by the CRR method.  Let $C=s_1
\vee \dots \vee s_k \vee \overline{x}$ be a clause selected by the
CRR method to constrain input assignments of the initial time frame.
Namely, $C$ removes every input assignment in which $s_i = 0,
i=1,\dots,k$ and $x=1$. Let $H_1,H_2, \dots,$ be range reduction
formulas computed with respect to clause $C$.  We will say that $H_i$
\tb{excludes} state \pnt{s} from $i$-th time frame if
$H_i(\pnt{s})=0$. 

If inputs of the initial time frame are not constrained by $C$, the
set of states of the counter reachable in $i$ transitions consists of
the $i+1$ states with values ranging from 0 to $i$. If $E$ is a trace
of $i$ transitions and $x$ is equal to 1 in $m$ time frames and equal
to 0 in $i-m$ time frames, the counter reaches a state \pnt{s} where
\vl{s} = $m$. If inputs of the initial time frame are constrained by
$C$, variable $x$ cannot have value 1 in the initial time frame. So
traces where $x$ is equal to 1 in \ti{every} time frame are
excluded. This means that the state \pnt{s} where \vl{s}=$i$ is
excluded by $H_i$. Note that every state \pnt{s} such that $0 \leq
\vl{s} < i$ is reachable in $i$ transitions by a trace where $x=0$ in
the initial time frame i.e. by an \ti{allowed} trace. Hence such a
state cannot be excluded by $H_i$ and so $H_i(\pnt{s})=1$.

Suppose that $n \geq d$. For every new range reduction formula $H_i$,
the CRR method checks if \Impl{\overline{H_i}}{P} holds. For the first
$d-1$ formulas $H_1,\dots,H_{d-1}$, this implication holds and so no
bad state is excluded. However, since $H_d$ excludes a state \pnt{s}
such that $\vl{s}=d$ and hence $P(\pnt{s})$=0, implication
\Impl{\overline{H_d}}{P} does not hold. At this point, the CRR method
reports that $P$ is broken by a trace of $d$ transitions.

Now, assume that $n < d$. Then \Impl{\overline{H_i}}{P} holds for all
formulas $H_1,\dots,H_n$. This means that formula $C$ is
$P^n$-equivalent.  That is constraining the inputs of the initial time
frame of $\xi$ with $C$ does not affect the answer to the question
whether $P$ holds for $n$ transitions. In general, one needs to add
many $P^n$-equivalent clauses to prove that a property holds for
system $\xi$ for $n$ transitions.  However, for our example, showing
that $C$ is $P^n$ equivalent is sufficient to finish the job.  Note
that only state \init~is possible in the initial time frame.  Due to
clause $C$, the value of $x$ in the initial time frame is fixed at
0. So only state \init~is possible in the next time frame that is the
same state as in the previous time frame.

At this point the CRR method stops to declare that $P$ holds for $n$
transitions.  In Subsection~\ref{subsec:intr_operation}, we gave two
conditions under which the CRR method claims that a property holds for
$n$ transitions. Our example employs the second condition. The set of
all traces of $m$ transitions where $m \leq n$ reduces to one trace
$L$ where the last state repeats a previously seen state of $L$.  In
our example, $L$ consists of two copies of state \init~and $m$ is
equal to 1.

\subsection{Comparison of the CRR method with other model checkers}
\label{subsec:deep_bugs}
In this subsection, we use our example to discuss the advantage of the
CRR method over other model checkers in the context of bug hunting.
To be concrete, let us assume that $d$=20,000 and one needs to check
if the property $P$ above holds for some $n$. We will assume that $n >
d$ and so $P$ does not hold.

To find a counterexample by BMC, one will have to generate formulas
$G_1,\dots,G_{20000}$ where satisfiability of $G_i$ means the
existence of a counterexample of $i$ transitions. Formula $G_i$
contains $i$ copies of the transition relation. So even if $\xi$ is
small, formulas $G_i$ grow too large to be solved efficiently by a
SAT-algorithm.

A model checker computing the set of states reachable in $k$
transitions $k=1,\dots,n$ or its over-approximation will have a
different kind of a problem.  Before searching for a counterexample of
20,000 transitions, such a model checker will have to prove that no
counterexample of at most 19,999 transitions exists. This requires
computing 19,999 sets of reachable states or their
over-approximations.

The computation above can be done efficiently only for particular
binary encodings of the values of the counter. Consider, for instance,
the usual binary encoding where the more significant a state bit is
the less frequently it toggles when the counter switches from the
current state to the next one.  In this case, there is a natural
ordering of state variables for which the set of states of the counter
reachable in $k$ transitions can be represented by a compact BDD. So a
BDD- based model checker will have no problem with finding a
counterexample. 

An IC3-like model checker that builds over-approximations of the set
of reachable states will also benefit of the encoding above. A key
operation of IC3 is to compute an inductive clause. To make this
computation efficient, state encoding should satisfy the following
property.  If there is a transition from state \pnt{s} to state
\pnt{s'}, the Hamming distance between \pnt{s} and \pnt{s'} should be
small. In the majority of transitions, the encoding above satisfies
this property. So, most likely, an IC3-like model checker will find a
counterexample efficiently.

As we mentioned above an abstract counter is meant to simulate a
sub-behavior of a sequential circuit, so, in general, no assumptions
about state encoding can be made. In this case, the size of a BDD
representing the set of states reachable in $k$ transitions can be
large no matter how variables are ordered. So finding a counterexample
by a BDD based model checker becomes inefficient. The same applies to an
IC3-like model checker. The reason is that generation of inductive
clauses becomes inefficient.

As we showed above, in our example, the CRR method builds range
reduction formulas $H_i$, $i=1,\dots,k$ that exclude only \ti{one}
state from $i$-th time frame.  That is to reach a bad state, the CRR
method needs to compute only one state per time frame as opposed to
computing the set of \ti{all} states reachable in $i$ transitions or
its over-approximation. For that reason, for our example, the CRR
method has very weak dependence on state encoding (if any). So,
arguably, it will be able to find a counterexample in cases where
other model checkers will fail.

%% file: p1qe.tex
\section{Partial Quantifier Elimination}
\label{sec:pqe}
In this section, we recall Quantifier Elimination (QE) and Partial QE
(PQE) the latter being a key operation of the CRR method. This section
is structured as follows. Subsection~\ref{subsec:qe_pqe_defs} defines
the QE and PQE problems. We introduce the notion of a noise-free
PQE-solver in Subsection~\ref{subsec:noise_free_alg}. This notion
plays an important role in reasoning about range reduction formulas
that we introduce in Section~\ref{sec:main_idea}. In
Subsection~\ref{subsec:qe_pqe_range}, we show that computing the range
of a circuit or reduction of the circuit range caused by input
constraints come down to QE and PQE respectively.

%
% Subsection: QE and PQE
%
\subsection{Quantifier elimination and partial quantifier elimination}
\label{subsec:qe_pqe_defs}
Let $G(X,Y)$ be a CNF formula.  We will call formula \prob{X}{G} an
{\boldmath $\exists\mi{CNF}$}.  The problem of \tb{Quantifier
  Elimination} (\tb{QE}) is to find a quantifier-free formula $H(Y)$
such that $H \equiv \prob{X}{G}$.

Let \prob{X}{F(X,Y) \wedge G(X,Y)} be an \ecnf{}. The problem of
\tb{Partial QE} (\tb{PQE}) is to find a quantifier-free formula
$F^*(Y)$ such that $F^* \wedge \prob{X}{G} \equiv \prob{X}{F \wedge
  G}$.  We will say that formula $F^*$ is obtained by taking $F$ out
of the scope of quantifiers in \prob{X}{F \wedge G}.

An obvious difference between PQE from QE is that the latter takes
\ti{the entire formula} $F \wedge G$ out of the scope of
quantifiers. Importantly, PQE can be dramatically simpler than QE
especially if formula $F$ is much simpler than $G$. In
Section~\ref{sec:bug_hunting} we show that computing range reduction
formulas comes down to an instance of the PQE problem. In this instance,
PQE is drastically simpler than QE because only a small part of the
formula is taken out of the scope of quantifiers.

% 
%  Subsection: Noise-free PQE-solver
%
\subsection{Noise-free PQE-solver}
\label{subsec:noise_free_alg}
Let $F^*(Y)$ be a solution to the PQE problem i.e.  $F^* \wedge
\prob{X}{G} \equiv \prob{X}{F \wedge G}$. Recall that $Y$ denotes the
set of free variables of \prob{X}{F \wedge G}. Let $C$ be a clause of
$F^*$ that is implied by $G$.  Then formula $F^* \setminus \s{C}$ is
also a solution to the same PQE problem.  That is $F^{**} \wedge
\prob{X}{G} \equiv \prob{X}{F \wedge G}$ where $F^{**} = F \setminus
\s{C}$.  One can think of clauses of $F^*$ implied by $G$ as
``noise''.  

Suppose that a clause $C$ of $F^*$ is not implied by $G$ but by adding
literals of variables from $Y \setminus \V{C}$ clause $C$ can be
extended to a clause implied by $G$. This can also be viewed as the
presence of some noise in $C$. We will say that a clause of $F^*$ is
\tb{noise-free} if the extension above does not exist.  We will call
$F^*$ \tb{a noise-free solution} if every clause of $F^*$ is
noise-free. We will call a PQE algorithm \tb{noise-free}
if it produces only noise-free solutions. A clause $C$, a solution $F^*$
and a PQE-algorithm that are not noise-free are called \tb{noisy}.

%
% Subsection: Relation of QE and PQE with computing range of a circuit
%
\subsection{Relation of QE and PQE to computing range of a circuit}
\label{subsec:qe_pqe_range}
Let $M(X,Y,Z)$ be a multi-output combinational circuit where $X,Y$ and
$Z$ specify input, intermediate and output variables of $M$
respectively.  In this subsection, we discuss QE and PQE in the context
of computing the range of $M$. Namely, we show that a) computing the
range of $M$ comes down to QE; b) PQE can be used to compute range
reduction caused by constraining inputs of $M$. 

In the two propositions below, we assume that $G(X,Y,Z)$ is a CNF
formula specifying circuit $M$ that is obtained by Tseitsin
transformations.
%
% Proposition: QE as range computation 
%
\begin{proposition}
\label{prop:full_range}
Let $R(Z)$ be a CNF formula such that $R \equiv \prob{W}{G}$ where
$W = X \cup Y$. (That is $R$ is a solution to the QE problem.)
Then the assignments satisfying $R(z)$ specify the range of $M$.
\end{proposition}

The proofs of all propositions are given in the appendix.

%
% Proposition: PQE as range reduction computation
%
\begin{proposition}
\label{prop:range_red}
Let $C(X)$ be a clause depending only on input variables of $M$.  Let
$H(Z)$ be a CNF formula such that $H \wedge \prob{W}{G} \equiv
\prob{W}{C \wedge G}$ where $W = X \cup Y$. (That is $H$ is a solution
to the PQE problem.)  Let $H$ and $H^*$  be  a noise-free and noisy 
solution respectively.  Then
\begin{enumerate}
\item The assignments \ti{falsifying} $H$ specify the range reduction
  in $M$ caused by excluding inputs falsifying $C$. That is $H(\pnt{z}) = 0$
  iff 
\begin{itemize}
   \item there is an input \pnt{x} for which circuit $M$ produces  output \pnt{z} 
   \item all inputs for which $M$ produces output \pnt{z} falsify $C$ 
\end{itemize}
\item \Impl{H^*}{H}
\end{enumerate}
\end{proposition}

%% file: d1efinitions.tex
\section{Notation and definitions}
\label{sec:defs}
Let $\xi$ be a state transition system with a transition relation
specified by a combinational circuit $N(S,X,Y,S')$. Here $S$ and
$S'$ are sets of present and next state variables, $X$ is the set of
combinational input variables, and $Y$ is the set of internal
combinational variables. Then $S \cup X$ (respectively $S'$) specify
the input variables (respectively output variables) of $N$. 
Let $T(S,X,Y,S')$ be a formula specifying $N$. Let $P(S)$ be a
property of $\xi$ and $I(S)$ be a formula specifying the set of
initial states of $\xi$. For the sake of simplicity, in the following
exposition we \ti{omit mentioning the variables of} $Y$. 
%
% Definition:  Trace
%
\begin{definition}
A complete assignment (\pnt{s},\pnt{x}) to variables of $(S,X)$ is
called an \tb{input pair}. We will refer to \pnt{s} (respectively
\pnt{x}) as a \tb{state} (respectively \pnt{X}\tb{-input}).  A
sequence \tr{s}{x}{0}{k} of input pairs is called a \tb{trace} of
$\xi$ if T(\pnt{s_i},\pnt{x},\pnt{s_{i+1}})=1, $0 \le i < k$.  If
$I(\pnt{s_0})=1$, this trace is called \tb{initialized}.
\end{definition}

%
% Definition:  State reachable by a trace
%
\begin{definition}
Let $E=\tr{s}{x}{0}{k}$ be a trace. Let \pnt{s_{k+1}} be the state to
which $\xi$ transitions under input pair \IP{s_k}{x_k}. We will call
\pnt{s_{k+1}} \tb{the state reachable by trace} \pnt{E}. We will also
say that \pnt{s_{k+1}} is reachable in $k+1$ transitions.
\end{definition}
%
% Definition: Bad and good states
%
\begin{definition}
Given a property $P(S)$ of system $\xi$, a state \pnt{s} is called
\tb{good} (respectively \tb{bad}) if $P(\pnt{s})=1$ (respectively
$P(\pnt{s})=0$).  Property $P$ is false for $\xi$ if there is an
initialized trace $E=\tr{s}{x}{0}{k}$
such that 
\begin{itemize}
\item every state \pnt{s_i} of $E$ is good $i=0,\dots,k$
\item the state \pnt{s_{k+1}} reachable by $E$ is bad
\end{itemize}
Trace $E$ is called a \tb{counterexample}. 
\end{definition}

%
% Numbering of time frames
%
\begin{definition}
We will index variables of system $\xi$ to distinguish between different time
frames. We will assume that numbering of time frames starts with 0.
We will refer to the time frame with index 0 as 
or the \tb{initial time frame}.

\end{definition}

%
% Definition: Traces excluded by formula $H$
%
\begin{definition}
\label{def:excl_traces}
Let $H(S,X)$ be a CNF formula that constrains the input pairs of the
system $\xi$ in the initial time frame. That is $H$ excludes every
initialized trace \Tr{s}{x}{0}{k} in which $(\pnt{s_0},\pnt{x_0})$
falsifies $H$. We will refer to such traces as \tb{excluded by
  formula} \pnt{H}. If the input pair $(\pnt{s_0},\pnt{x_0})$ of an
initialized trace satisfies $H$, this trace is said to be \tb{allowed}
by $H$.
\end{definition}

%
% Definition: P^n-equivalence
%
\begin{definition}
Let $P$ be a property of system $\xi$. Let $C(S,X)$ be a clause
excluding input pairs of the initial time frame. Suppose that $P$
holds for system $\xi$ for $n$ transitions iff the set of traces
\ti{allowed} by clause $C$ contains a counterexample of length at most
$n$. We will say that the system constrained by $C$ is
\pnt{P^n}\tb{-equivalent} to the original system $\xi$. Informally,
$P^n$-equivalence means that discarding the traces of $\xi$ excluded
by $C$ does not eliminate \ti{all} counterexamples of length at most
$n$ (if any). We will call clause $C$ preserving $P^n$-equivalence of
system $\xi$ a \pnt{P^n}\tb{-equivalent clause}.
\end{definition}

%% file: m0ain_idea_mc_ccr.tex
\section{Model Checking By CRR}
\label{sec:main_idea}
In this section, we give an introduction into model checking by
Computing Range Reduction (CRR).  First, we outline the main idea in
Subsection~\ref{subsec:main_idea}. Then, in
Subsection~\ref{subsec:compl_alg}, we give a high-level 
description of a model checker based on CRR.
%
% Subsection: main idea
%
\subsection{Main idea}
\label{subsec:main_idea}
Let $\xi$ be a system introduced in Section~\ref{sec:defs} and $P$ be
a property of $\xi$.  We will assume that \Impl{I}{P}, that is all
initial states satisfy property $P$.  Let $C(S,X)$ be a clause
specified in terms of input variables of circuit $N$ above such that
\Impl{\overline{C}}{I}. Suppose that we use $C$ to exclude traces as
described in Definition~\ref{def:excl_traces}.  Suppose that a state
\pnt{s} of $\xi$ is reachable in $i$ transitions only by traces
excluded by $C$. This means that if one discards the traces excluded
by $C$, state \pnt{s} becomes unreachable in $i$ transitions. Such
states are specified by range reduction formulas defined below.

%
% Definition: Range Reduction formulas
%
\begin{definition}
\label{def:rr_formulas}
The result of using clause $C$ to exclude traces of $\xi$ of length at
most $n$ can be characterized by a set of formulas $H_1,\dots,H_n$
defined as follows. The value of $H_i(\pnt{s})$ is equal to 0 iff
\begin{itemize}
\item \pnt{s} is reachable in $i$ transitions
\item all traces of length $i$ that reach \pnt{s} are excluded
  by $C$.
\end{itemize}
We will call $H_i$ a \tb{range reduction formula}.  We will say that
state \pnt{s} \tb{is excluded by} \pnt{H_i} if $H_i(\pnt{s})=0$.
\end{definition}

Model checking by Computing Range Reduction (CRR) is based on
the following four observations. The first observation is that formula
$H_i$ specifies a \ti{reduction in the range} of a circuit obtained by
the composition of $i$ circuits $N$. Such a circuit describes
the traces  of $i$ transitions. The change of range described by $H_i$
is caused by discarding traces excluded by clause $C$.  Using
Proposition~\ref{prop:range_red}, one can compute such range
reductions by a PQE solver.

The second observation is that one can use range reduction formulas to
find a counterexample. Suppose that \Nmpl{\overline{H_i}}{P} i.e.
there is a state \pnt{s} such that $H_i(\pnt{s})=0$ and
$P(\pnt{s})=0$. This means that by discarding the traces excluded by
clause $C$, one excludes a \ti{bad} state \pnt{s} from the set of
states reachable in $i$ transitions. This implies that there is a
counterexample formed by a trace \tr{s}{x}{0}{i} excluded by $C$
leading to a bad state.  This trace can be easily recovered from
$H_1,\dots,H_i$ and clause $C$ by $i+1$ SAT-checks. In more detail,
bug hunting by CRR is described in Section~\ref{sec:bug_hunting}.

The third observation is as follows. As mentioned in the introduction,
formula $H_i$ specifies the difference between sets $A_i$ and
$B_i$. Set $A_i$ consists of the states that can be reached by traces
of length $i$ that are excluded by $C$. Set $B_i$ is a subset of $A_i$
that consists of the states that are also reachable by traces of
length $i$ that are \ti{allowed} by $C$. The set $A_i \setminus B_i$
represented by $H_i$ consists of the states that can be reached
\ti{only} by traces of length $i$ excluded by $C$.  This set can be
\ti{very small} even when sets $A_i$ and $B_i$ are huge. In
Section~\ref{sec:experiments}, we give some experimental evidence to
support this conjecture. Informally, this means that a model checker
based on CRR can find a bug by examining a very small number of
states.

The fourth observation is that one may not need to compute all $n$
range reduction formulas $H_i$ to prove that clause $C$ is
$P^n$-equivalent. Suppose, for example, that formula $H_i$ is empty
where $i < n$. That is $H_i \equiv 1$ (and hence $H_i$ cannot exclude a
bad state).  Then every formula $H_j$, $i < j \leq n$ is also empty.

%
% Subsection: high-level description of a model checker based on CRR
%
\subsection{High-level description of a model checker based on CRR}
\label{subsec:compl_alg}
In this subsection, we give a high-level explanation of how one can
build a model checker based on CRR that checks if a property $P$
holds for $n$ transitions. A detailed description of an instance of
such a model checker is given in Section~\ref{sec:mc_crr}.

%
% Definition: An initial input pair
%
\begin{definition}
Let \IP{s}{x} be an input pair where \pnt{s} is an initial
state. Then we will call this pair an \tb{initial input pair}.
\end{definition}

Suppose that one excludes the initial input pairs of $\xi$ as follows.
First, an initial input pair (\pnt{s},\pnt{x}) is picked. Then a
clause $C(S,X)$ falsified by (\pnt{s},\pnt{x}) is generated such that
\Impl{\overline{C}}{I}. After that, range reduction formulas
$H_1,\dots,H_n$ are computed with respect to clause $C$.  If a formula
$H_i$ does not imply $P$ i.e. $H_i$ excludes a bad state, a
counterexample is generated. Otherwise, one proves that $C$ is a
$P^n$-equivalent clause. After that, $C$ is added to a formula $Q$
that accumulates all $P^n$-equivalent clauses generated so far to
exclude initial input pairs. Initially, $Q$ is empty.

Then one picks an initial input pair (\pnt{s},\pnt{x}) that satisfies
$Q$. This guarantees that this a new initial input pair. A new clause
$C$ falsified by this input pair is generated and a new set of range
reduction formulas is generated with respect to clause $C$.  The
process of elimination of initial input pairs goes on until either a
counterexample is generated or all initial input pairs but one are
excluded.  The reason why the last initial input pair is not excluded
is as follows. In Subsection~\ref{subsec:main_idea} we mentioned that
$H_i$ represents the difference of sets $A_i$ and $B_i$ where $B_i$ is
a subset of $A_i$.  The larger the set $B_i$, the smaller the set $A_i
\setminus B_i$ that $H_i$ represents. The size of the set $B_i$
depends on the number of traces of length $i$ that are \ti{allowed} by
clause $C$.  If the last initial input pair is eliminated by a clause
$C$, then \ti{no trace of length} $i$ is allowed by $C$. In this case,
the set $B_i$ is empty and the CRR method essentially reduces to
reachability analysis where set $A_i$ grows uncontrollably.

Let \IP{e_0}{d_0} be the initial input pair that still satisfies $Q$.
This means that every remaining counterexample (if any) starts with
the input pair \IP{e_0}{d_0}. Let \pnt{e_1} denote the state to which
$\xi$ transitions to under input \IP{e_0}{d_0}. Obviously, the traces
of $\xi$ allowed by $Q$ go through state \pnt{e_1}.  This means that
the original system $\xi$ with initial states specified by formula $I$
is $P^n$-equivalent with $\xi$ that has only one initial state equal to
\pnt{e_1}. This also means that the initial time frame can be
discarded.

One can use the same procedure of building formula $Q$ that excludes
the initial input pairs of the modified $\xi$. This initial input
pairs are of the form \IP{e_1}{x} where \pnt{x} is a complete
assignment to variables of $X$.  The procedure described above can be
used to eliminate all initial input pairs but an input pair
\IP{e_1}{d_1}.  This means that every remaining counterexample of the
original system $\xi$ has to start with \IP{e_0}{d_0},\IP{e_1}{d_1}.

The procedure of elimination of initial input pairs has the following
three outcomes. Suppose that the first $k$ time frames of $\xi$ have
collapsed to trace \tr{e}{d}{0}{k-1}. Let \pnt{e_k} denote the state
to which $\xi$ transitions under input \IP{e_{k-1}}{d_{k-1}}. The
first outcome is as follows. Suppose that when eliminating an initial
input pair \IP{e_k}{x} one of the range reduction functions excludes a
bad state \pnt{e_{m+1}}. Then one can build a trace \tr{e}{d}{k}{m}
leading to \pnt{e_{m+1}}. This trace can be extended to trace
\tr{e}{s}{0}{m} that is a counterexample of the initial system $\xi$.

The second outcome is that \pnt{e_k} repeats a state \pnt{e_i}, $i
\leq k$.  This means that the procedure of excluding initial input
pairs above will be reproducing the same states between \pnt{e_i} and
\pnt{e_k}.  So no counterexample of length at most $n$ breaking
property $P$ exists and hence $P$ holds for $n$ transitions.

The third outcome is that the first $n$ time frames are collapsed to a
trace \tr{e}{d}{0}{n-1} where all states \pnt{e_i} are good and different
from each other. This means that property $P$ holds for $n$
transitions.

%% file: isol_publ_cex.tex
\section{Isolated And Public Traces}
\label{sec:isol_publ_traces}
In this section, we classify the traces excluded by a clause $C$ into
two sets: isolated traces and public traces.  The importance of such
classification is as follows. First, as we show in
Section~\ref{sec:bug_hunting}, one can use CRR to efficiently find
isolated counterexamples. Second, as we prove in
Proposition~\ref{prop:p_equiv_clause} below, if $C$ does not exclude
an isolated counterexample of length at most $n$ disproving property
$P$, then $C$ is a $P^n$-equivalent clause.
%
% Definition: isolated trace 
%
\begin{definition}
\label{def:isol_trace}
Let $\xi$ be a state transition system.  Let $C(S,X)$ be a clause such
that \Impl{\overline{C}}{I}. Let $E$ denote an initialized trace
\tr{s}{x}{0}{m} that is excluded by $C$.  We will call $E$
\tb{isolated with respect to clause} \pnt{C} if no state \pnt{s_i},$i
> 0$ of $E$ can be reached by a trace of length $i$ allowed by
$C$. Otherwise, $E$ is said to be \tb{public with respect to clause}
\pnt{C}.
\end{definition}

%
% Proposition: about isolated traces
%
\begin{proposition}
\label{prop:isol_traces}
Let $C(S,X)$ be a clause such that \Impl{\overline{C}}{I}. Let
$H_1,\dots,H_m$ be range reduction formulas computed with respect to
clause $C$. Let $E$ denote an initialized trace \tr{s}{x}{0}{m} such
that
\begin{itemize}
\item \IP{s_0}{x_0} falsifies $C$ i.e. $E$ is excluded by $C$
\item \IP{s_i}{x_i} falsifies $H_i$, $i=1,\dots,m$.
\end{itemize} 
 Then $E$ is isolated with respect to $C$.
\end{proposition}

%
% Definition: isolated counterexample
%
\begin{definition}
Let $\xi$ be a system with property $P$.  Let $C(S,X)$ be a clause
such that \Impl{\overline{C}}{I}. Let $E$ denote a a counterexample
\tr{s}{x}{0}{m}. We will say that $E$ is a \tb{counterexample isolated
  (or public) with respect to clause} \pnt{C} if trace $E$ is isolated
(respectively public) with respect to $C$.
\end{definition}
%
% Proposition: P-equivalence of a clause
%
\begin{proposition}
\label{prop:p_equiv_clause}
Let $\xi$ be a system with property $P$. Let $C(S,X)$ be a clause such
that \Impl{\overline{C}}{I}. Assume that $C$ does not exclude any
counterexample of length at most $n$ isolated with respect to $C$.
Then $C$ is a $P^n$-equivalent clause.
\end{proposition}

%% file: bug_hunting.tex
\section{Bug Hunting By CRR}
\label{sec:bug_hunting}
In this section, we describe how bug hunting is done by CRR.  In
Subsection~\ref{subsec:bh_noise_free}, we discuss the construction of
range reduction formulas by a noise-free PQE solver.  We show that
such a PQE-solver can \ti{prove the existence} of a bug without
generation of an explicit counterexample. This is done by just showing
that excluding initial input pairs by a clause $C$ leads to excluding
a trace of length $k$ leading to a bad state.  In
Subsection~\ref{subsec:bh_noisy}, we describe building range reduction
formulas by a noisy PQE-solver. We show that in this case, one has to
build a counterexample \ti{explicitly}.

%
% Subsection: Bug hunting by a noise free PQE solver
%
\subsection{Bug hunting with a noise-free PQE solver}
\label{subsec:bh_noise_free}

%
%  Proposition: CRR by a noise-free PQE solver
%
\begin{proposition}
\label{prop:noise_free_pqe}
Let $\xi$ be a state transition system with property $P$. Let $C(S,X)$
be a non-empty clause such that \Impl{\overline{C}}{I}.  Let $H_0$
denote formula equal to $C$. Let formulas $H_1,\dots,H_n$ be obtained
recursively as follows.  Let $\Phi_0$ denote formula equal to $I$.
Let $\Phi_i,~~0 < i \leq n$ denote formula $I \wedge H_0 \wedge T_0
\wedge \dots \wedge H_{i-1} \wedge T_{i-1}$. Here $T_j =
T(S_j,X_j,S_{j+1})$ where $S_j$ and $X_j$ are state and input
variables of $j$-th time frame.  Formula $H_{i+1}$ is obtained by
taking $H_i$ out of the scope of quantifiers in formula \prob{W}{H_i
  \wedge T_i \wedge \Phi_i} where $W = S_0 \cup X_0 \cup \dots \cup
S_i \cup X_i$.  That is $H_{i+1} \wedge \prob{W}{T_i \wedge \Phi_i}
\equiv \prob{W}{H_i \wedge T_i \wedge \Phi_i}$.  Then formulas
$H_1,\dots,H_n$ are range reduction formulas.
\end{proposition}

The fact that range reduction formula $H_i$ excludes only reachable
states guarantees that if \Nmpl{\overline{H}_i}{P} then a
counterexample exists.  As we show below this is not true when a noisy
PQE solver is used.

%
% Subsection: Bug hunting by a noisy PQE solver
%
\subsection{Bug hunting with a noisy PQE solver}

\label{subsec:bh_noisy}
%
% Proposition: CRR by a noisy PQE solver
%
\begin{proposition}
\label{prop:noisy_pqe}
Let $H^*_i, i=0,\dots,n$ be formulas obtained as described in
Proposition~\ref{prop:noise_free_pqe} with only one exception. A \ti{noisy}
PQE-solver is used to obtain $H^*_{i+1}$ by taking $H^*_i$ out of the
scope of quantifiers in \prob{W}{H^*_i \wedge T_i \wedge \Phi^*_i}.
Here $\Phi^*_0 = I$, $H^*_0=C$ and $\Phi^*_i = I \wedge H^*_0 \wedge
T_0 \wedge \dots \wedge H^*_{i-1} \wedge T_{i-1}$ for $i < 0 \leq n$.
Then \Impl{H^*_i}{H_i} holds where $H_i,i=1,\dots,n$ are range
reduction formulas.
\end{proposition}

Proposition~\ref{prop:noisy_pqe} suggests that a noisy PQE-solver, in
general, builds a formula $H^*_i$ that over-approximates the set of
states for which a correct range reduction formula $H_i$ evaluates to
0.  For that reason we will refer to $H^*_i$ as an \tb{approximate
  range reduction formula}. Since $H^*_i$ is not logically equivalent
to $H_i$, the former can exclude states that are not reachable by
$\xi$.  So if \Nmpl{\overline{H^*_i}}{P} for some state \pnt{s_i}, one
needs to check if \pnt{s_i} is reachable from an initial state. This
can be done as follows. First, a state \pnt{s_{i-1}} from which there
is a transition to \pnt{s_i} is searched for. If such a state exists,
then a state \pnt{s_{i-2}} from which there is a transition to
\pnt{s_{i-1}} is searched for and so on. This process results either
in finding a counterexample reaching state \pnt{s_i} or deriving a
clause falsified by \pnt{s_i}. The latter means that \pnt{s_i} is
unreachable.
\subsection{Building range reduction formulas incrementally}
\label{subsec:rr_form_increm}
In the previous subsection, we showed that a range reduction formula
$H_{i+1}$ can be obtained by taking $H_i$ out of the scope of
quantifiers in \prob{W}{H_i \wedge T_i \wedge \Phi_i}.  Note that
formula $\Phi_i$ contains $i-1$ copies of the transition relation and
so gets very large as $i$ grows. Fortunately, in general, one only
needs a small set of time frames preceding the time frame $i$ to
derive the clauses of $H_{i+1}$. This makes derivation of $H_{i+1}$
\ti{local}.

The reason for derivation of $H_{i+1}$ to be local is as follows.
Solving the PQE problem comes down to generating a set of clauses
depending on free variables of \prob{W}{H_i \wedge T \wedge \Phi_i}
that makes the clauses of $H_i$ redundant in \prob{W}{H_i \wedge T
  \wedge \Phi_i}. In~\cite{tech_rep_pqe}, we introduced a PQE-solver
called \PQE that implements this strategy.  To solve the PQE problem,
\PQE maintains a set of clauses to be Proved Redundant. We will refer
to a clause of this set as a PR-clause. Originally, the set of
PR-clauses consists of the clauses of $H_i$. A resolvent clause $C$
that is a descendant of $H_i$ also becomes a PR clause and so needs to
be proved redundant. The only exception is the case when $C$ depends
only of free variables of \prob{W}{H_i \wedge T \wedge \Phi_i} i.e. on
variables of $S_{i+1}$. Then $C$ is just added to $H_{i+1}$.

\PQE uses branching to first prove redundancy of PR-clauses in
subspaces. Then it merges the results of different
branches. Importantly, \PQE backtracks as soon as all PR-clauses are
proved redundant in the current subspace. This means that \PQE needs
clauses of $\Phi_i$ corresponding to $j$-th time frame where $j < i$
only if there is a PR-clause that contains a variable of $j$-th time
frame.  \PQE produces a new PR-clause $C''$ obtained from another
PR-clause \ti{only} if it cannot prove redundancy of $C'$ in the
current subspace.  Generation of $C''$ can be avoided if the current
formula has a non-PR clause that subsumes $C'$ in the current
subspace. For instance, one could prevent the appearance of clause
$C''$ containing a variable of $j$-th time frame by exploiting non-PR
clauses that depend on variables of $S_{j+1}$ and implied by formula
$\Phi_{j+2}$.  Such clauses could have been derived when building
range reduction formula $H_{j+2}$ by taking $H_{j+1}$ out of the scope
of quantifiers.

So, making computation of $H_{i+1}$ local comes down to preventing the
appearance of PR-clauses containing variables of time frames that are
far away from the $i$-th time frame. This is achieved by re-using
clauses derived when building range reduction formulas $H_j$, $j \leq
i$.

%% file: p2roving_corr.tex
\section{Identification Of $P$-Equivalent Clauses}
\label{sec:p_equiv_clauses}
In this section, we describe two cases where one can prove that a
clause $C$ is $P^n$-equivalent. We assume here that computation of
range reduction formulas is performed by a noisy
PQE-solver. Proposition~\ref{prop:trivial_case} describes the case
where one needs to compute all formulas $H^*_i$, $i=1,\dots,n$. The
case where this is not necessary is addressed by
Proposition~\ref{prop:empty_rr_form}.
%
% Trivial case of a P^n-equivalent clause
%
\begin{proposition}
\label{prop:trivial_case}
Let $\xi$ be a system with property $P$. Let $C(S,X)$ be a clause such
that \Impl{\overline{C}}{I}. Let $H^*_1,\dots,H^*_n$ be approximate
range reduction formulas computed with respect to clause $C$ by a
noisy PQE solver.  Suppose that no formula $H^*_i$,~$i=1,\dots,n$
excludes a reachable bad state \pnt{s}. Then clause $C$ is
$P^n$-equivalent.
\end{proposition}

%
% P^n-equivalent clause when H_i is empty
%
\begin{proposition}
\label{prop:empty_rr_form}
Let $\xi$ be a system with property $P$. Let $C(S,X)$ be a clause such
that \Impl{\overline{C}}{I}. Let $H^*_1,\dots,H^*_i$ be approximate
range reduction formulas computed with respect to clause $C$ by a
noisy PQE solver. Suppose that every bad state excluded by $H^*_j$, $1
\leq j < i$ is unreachable. Suppose that every state (bad or good)
excluded by $H^*_i$ is unreachable. Then clause $C$ is
$P^n$-equivalent for any $n > 0$.
\end{proposition}

Proposition~\ref{prop:empty_rr_form} suggests that one can declare a
clause $C$ $P^n$-equivalent for an arbitrary $n$ if formula $H^*_i$
does not exclude any reachable states.  Let us consider the following
three cases. The first case is that $H^*_i$ is empty i.e. $H^*_i
\equiv 0$.  The second case, is that $H^*_i$ excludes only bad states
i.e. \Impl{\overline{H^*_i}}{\overline{P}}. The third case occurs when
\Nmpl{\overline{H^*_i}}{\overline{P}} i.e. when $H^*_i$ excludes good
states.

From the viewpoint of performance, it makes sense to check if every
state excluded by $H^*_i$ is unreachable only in the first and second
cases. In the first case, no state is excluded by $H^*_i$ and so no
extra work needs to be done to apply
Proposition~\ref{prop:empty_rr_form}. In the second case, one needs to
perform only one extra check to verify if
\Impl{\overline{H^*_i}}{\overline{P}} holds. Checking if a \ti{bad}
state excluded by $H^*_i$ is reachable has to be done anyway to
guarantee that no counterexample excluded by clause $C$ is overlooked.
On the other hand, if \Impl{\overline{H^*_i}}{\overline{P}} does not
hold i.e. if the third case above occurs, the amount of extra work one
has to do can be very high. This is because $H^*_i$ can exclude a
large number of unreachable \ti{good} states.

%% file: c3omparison.tex
\section{Comparison Of Model Checking By CRR With Other Approaches}
\label{sec:comparison}
CRR is essentially a method of computing an \ti{under-approximation}
of the set of reachable states.  In this section, we compare the CRR
method with Bounded Model Checking (BMC) and with methods based on
computing an \ti{Over-approximation} of Reachable States. We will
refer to the latter as ORS-methods.  We will assume that the precise
computation of a set of reachable states is just a special case of its
over-approximation.

In Section~\ref{sec:example}, we already made some comparison of the
CRR method with other model checkers on a simple example. Here, we
continue this work in the general case. Since, in this paper, we
emphasize the great potential of using CRR in bug hunting we compare
the CRR method with BMC and ORS-methods in the context of generation
of counterexamples.  In Subsection~\ref{subsec:bh_ORS}, we compare the
CRR method and ORS-methods. In Subsection~\ref{subsec:bh_BMC} we
relate the bug hunting of the CRR method with that of BMC. For the
sake of simplicity, in this section, we assume that the CRR method
employs a noise-free PQE-solver.

%
% Subsection: Bug hunting by CRR and ORS-methods
%
\subsection{CRR and ORS-methods}
\label{subsec:bh_ORS}
The difference between the CRR method and ORS-methods as far as bug
hunting is concerned is that the ORS-methods look for a bug in a
breadth-first manner.  In particular, they try to find the shortest
possible counterexample. For example, IC3 first makes sure that a set
of invariants is met that guarantees that no counter-example of length
$n$ exists before it increments $n$ by 1.  The reason for such
strategy is that the number of states reachable in $n$ transitions
exponentially grows in $n$. This cripples the performance of model
checkers that compute the set of reachable sets precisely.  The best
model checkers like IC3 address this problem by over-approximating the
set of reachable states. However, finding bugs in a breadth-first
manner may render inefficient even successful ORS-methods like IC3.

In contrast, the CRR-method looks for a bug in a \ti{depth-first
  manner}.  Given a clause $C(S,X)$ and a number $n$ it computes a set
of range reduction formulas $H_1,\dots,$. This computation goes on
until a formula $H_i$, $i \leq n$ excludes a bad state or $C$ is
proved $P^n$-equivalent. A remarkable fact here is that the CRR method
can \ti{lock onto} a counterexample  (that is isolated with respect
$C$) by computing a drastic under-approximation of the set of
reachable states. By locking onto a counterexample $E=\tr{s}{x}{0}{m}$ we
mean generation of sets $D_i$, $i=0,\dots,m$ such that $D_0 \times
\dots \times D_m$ contains the tuple \TR{s}{0}{m}.

As we mentioned in Subsection~\ref{subsec:main_idea}, the set of
states excluded by formula $H_i$ can be represented as the difference
of sets $A_i$ and $B_i$. Here $A_i$ is the set of states reachable by
traces excluded by clause $C$ and $B_i$ is the subset of $A_i$
consisting of the states that are also reachable by traces allowed by
$C$.  Notice that to lock onto counterexample $E$, an ORS-method would
have to compute sets $A_1,\dots,A_m$ or their over-approximation.  The
CRR method locks onto $E$ by computing \ti{only sets} $A_1 \setminus
B_1, \dots, A_m \setminus B_m$ that can be drastically smaller.

%
% Bug hunting by CRR and BMC
%
\subsection{CRR and BMC}
\label{subsec:bh_BMC}
Similarly to ORS-methods, BMC searches for a counterexample in a
breadth-first manner. First BMC searches for a counterexample of
length 1.  If no such counterexample exists, BMC searches for a
counterexample of length 2 and so on. So the main difference of the
CRR method from BMC  is that the former searches for a counterexample in a
depth-first manner.

To check if a counterexample of length $i$ exists, BMC tests the
satisfiability of formula $G_i$ equal to $I \wedge T_0 \dots \wedge
T_{i-1} \wedge \overline{P}$.  The size of $G_i$ is linear in the
number of time frames.  So the reach of BMC is typically limited to
counterexamples of length 100-200. In theory, the CRR method has to
deal with formula $\Phi_i$ whose size is linear in $i$.  However, as
we conjectured in Subsection~\ref{subsec:rr_form_increm}, the
computation of range formula $H_{i+1}$ that involves formula $\Phi_i$
can be made local. In this case, only clauses of a small number of
time frames preceding time frame $i$ are employed.  So the CRR method
can potentially find very long counterexamples.

One more important advantage of the CRR method over BMC is that it can
derive clauses that are hard or even impossible to derive by a regular
SAT-solver~\cite{tech_rep_pqe}. Suppose, for example, that a clause
$C$ is proved $P^i$-equivalent. If $C$ eliminates a counterexample of
length $i$ it is not implied by the formula $G_i$ that BMC checks for
satisfiability. (Because $C$ eliminates an assignment satisfying
$G_i$.) So clause $C$ cannot be derived by a resolution based
SAT-solver from $G_i$.  Adding $C$ to $G_i$ only preserves the
satisfiability of the latter.

Importantly, the CRR method derives a $P^i$-equivalent clause $C$
differently from BMC even if $C$ is implied by $G_i$. We assume here
that the PQE-solver used by the CRR method to compute range
reduction formulas employs the machinery of
D-sequents~\cite{fmcad13}. Then, such a PQE-solver can produce clauses
obtained by non-resolution derivation. An example of a clause obtained
by non-resolution derivation is a blocked clause
~\cite{blocked_clause}.  Adding clauses obtained by non-resolution
derivation allows one to get proofs that are much shorter than those
based on pure resolution.  For example, in~\cite{gen_ext_resol} it was
shown that extending resolution with a rule allowing to add blocked
clauses makes it exponentially more powerful.

%% file: m1c_crr.tex
\section{Description of \MC}
\label{sec:pseudocode}
In this section, we describe a model checker called \mc that is based
on CRR.  To make this description simpler we omitted some obvious
optimizations.  For example, the current version of \mc discards
approximate range reduction formulas $H^*_i$,~$i=1,\dots,n$ computed
with respect to a clause $C$ after $C$ is proved $P^n$-equivalent.
Only clause $C$ itself is kept and re-used when a new clause $C'$ is
checked for being $P^n$-equivalent. In reality, formulas $H^*_i$ can
be re-used as well. The same applies to formulas $U_i$,~$i=1,\dots,n$
generated when eliminating unreachable bad states falsifying
$H^*_i$. In the current version of \mc, these formulas are discarded
after $C$ is proved $P^n$-equivalent. In reality, they can be re-used
when checking $P^n$-equivalence of a new clause $C'$.
\label{sec:mc_crr}
%\input{m2c_defs.tex}
%
% Subsection: Algorithm description
%
\subsection{Main procedure}
\input{m3c.fig} The pseudo-code of \mc is given in
Figure~\ref{fig:mc_crr}.  \mc accepts a state transition system $\xi$
described by a circuit $N$ and predicates $I$ and $P$.  $N$ specifies
a transition relation and predicates $I$ and $P$ specifies initial
states and the property to verify. \mc also accepts parameter $n$
informing the model checker that one needs to check if $P$ holds for
$n$ transitions. The main parts of the code are separated by the
dotted lines. \mc starts with generating formula $T$ specifying the
transition relation represented by circuit $N$ and initializing some
variables (lines 1-3).  As we mentioned earlier, \mc reduces the set
of all traces of length at most $n$ to one trace. This is done by
keeping only one input pair per time frame processed by \MC.  The set
of these input pairs is accumulated in variable \ti{Trace}. Variable
\ti{States} collects all the states of the trace stored in
\ti{Trace}. Variables \ti{Trace} and \ti{States} are initialized to an
empty set (lines 2-3).

The main computation of \mc takes place in a \ti{while} loop (lines
4-14). The body of the loop consists of two parts. In the first part
(lines 4-8) \mc generates an input pair \IP{s}{x} that \mc does not
exclude from the current initial time frame (line 5).  Then \mc checks that
property $P$ holds for state \pnt{s'} to which system $\xi$
transitions under input pair \IP{s}{x} (lines 6-7). If \pnt{s'} breaks
$P$ then \mc returns $\mi{Trace} \cup \s{\IP{s}{x}}$ as a
counterexample. Otherwise, \mc generates $A(S,X)$, the longest clause
falsified by \IP{s}{x}. For every clause $C$ generated by procedure
\ti{ConstrTimeFrame} (line 9) to constrain input pairs of the initial
time frame, \Nmpl{C}{A} holds.  This guarantees that $C$ does not
exclude the input pair \IP{s}{x}.

\mc starts the second part of the loop (lines 9-14) calling procedure
\ti{ConstrTimeFrame}. This procedure tries to exclude all the input
pairs of the current time frame but \IP{s}{x}. If \ti{ConstrTimeFrame}
fails to do this, it returns a trace $E$ for which $P$ fails. This
trace does not include the time frames already processed by \mc. For
that reason, to form a counterexample for the original system $\xi$
one needs to take the union of \ti{Trace} and $E$. If
\ti{ConstrTimeFrame} succeeds, then either property $P$ holds for $n$
transitions or every counterexample contains the input pair \IP{s}{x}
selected in line 5 that has not been eliminated.  In this case, \mc
updates sets \ti{Trace} and \ti{States} by adding \IP{s}{x} and
\pnt{s} respectively.  Then \mc checks if the state \pnt{s'} to which
$\xi$ transitions under input \IP{s}{x} is already in the set
\ti{States} (line 13). (Recall that \ti{States} contains all the states
of \ti{Trace}.) If it is, then property $P$ holds for $n$ transitions
because the part of \ti{Trace} between two copies of state \pnt{s'}
can be repeated.  Finally, \mc eliminates the current time frame by
making \pnt{s'} the new initial state (line 14). So the second time
frame of system $\xi$ with initial states $I$ becomes the new initial
time frame of $\xi$ with initial state \pnt{s'}.

%
% Subsection: constraining a time frame
%
\subsection{Constraining a time frame}
\input{c1onstr_time_frame.fig}

%
The pseudo-code of procedure \ti{ConstrTimeFrame} is shown in
Figure~\ref{fig:constr_time_frame}.  The objective of this procedure
is to exclude all input pairs of the initial time frame but the input
pair falsifying clause $A$.  The set of generated $P^n$-equivalent
clauses is accumulated in formula $G$ that is initially empty.
Computation is performed in a \ti{while} loop. First,
\ti{ConstrTimeFrame} generates a new input pair \IP{s}{x} to exclude.
This input pair is formed as an assignment satisfying $G \wedge A$
(line 3).  If $G \wedge A$ is unsatisfiable, then $G$ has excluded all
the input pairs but the input pair falsifying $A$. In this case,
\ti{ConstrTimeFrame} returns \ti{nil} meaning that no counterexample
is found (line 4).

If a new input pair to exclude \IP{s}{x} is found,
\ti{ConstrTimeFrame} generates a clause $C$ that is falsified by
\IP{s}{x} (line 5).  Clause $C$ is constructed in such a way that it
does not exclude the input pair falsifying $A$.  Then procedure
\ti{CompRrForm} is called that computes range reduction formulas with
respect to clause $C$ (line 6). If \ti{CompRrForm} returns a
counterexample (line 7) it means that a bad reachable state was
excluded by one of the range reduction formulas. Otherwise, $C$ is
added to $G$ as a $P^n$-equivalent clause.
%
% Subsection: range reduction formulas
%
\subsection{Computing range reduction formulas}
\input{c4omp_rr_formulas.fig}

%
Procedure \ti{CompRrForm} computing a set \rr=\s{H^*_0,\dots,} of
range reduction formulas is shown in Figure~\ref{fig:comp_rr_form}. We
assume that \mc employs a noisy PQE-solver. So $H^*_i$ are approximate
range reduction formulas.  For the sake of simplicity, in this section
we will refer to $H^*_i$ as just range reduction formulas. Formula
$H^*_i$ specifies states that become unreachable in $i$ transitions due
to excluding the input pairs of the initial time frame that falsify
clause $C$. Formula $H^*_0$ is just equal to clause $C$ (line 2). The
formulas $H^*_i$, $0 < i \leq n$ are initialized to an empty set of
clauses. \ti{CompRrForm} also forms the set \ur=\s{U_1,\dots,U_n}
(line 3).  Formula $U_i$ is meant to eliminate the bad states
falsifying formula $H^*_i$ that are unreachable. In other words, $U_i$
is meant to make up for the fact that \mc uses a noisy PQE-algorithm.
\ti{CompRrForm} also initializes set $W$ and formula $\Phi^*$ (lines 4-5).
They are used in formulating PQE problems to be solved (line 7).

The main computation is performed in a \ti{while} loop (lines 6-13). First,
a range reduction formula $H^*_{j+1}$ is computed by taking $H^*_j$ out of
the scope of quantifiers in formula \prob{W}{H^*_j \wedge T_j \wedge
  \Phi^*} (line 7).  Here $W$ is equal to $S_0 \cup X_0 \cup \dots \cup
S_j \cup X_j$ for $j \geq 0$ and $\Phi^*$ is equal to $I \wedge H^*_0
\wedge T_0 \dots \wedge H^*_{j-1} \wedge T_{j-1}$ for $j > 0$.  Then
\ti{CompRrForm} checks if $H^*_{j+1}$ excludes a bad state (line 8).  If
it does, then procedure \ti{ElimBadStates} is called (line 9) to check
if a bad state excluded by $H^*_{j+1}$ is reachable from an initial
state. If so, then \ti{ElimBadStates} returns a trace $E$ leading to such
a bad state (line 10).  In the process of checking if bad states
excluded by $H^*_{j+1}$ are reachable, \ti{ElimBadStates} updates
formulas $U_i$, $i = 1,\dots,j+1$ by adding new clauses eliminating
unreachable bad states.

After all bad states excluded by $H^*_{i+1}$ are eliminated by
clauses of $U_{j+1}$, \ti{CompRrForm} checks the condition of
Proposition~\ref{prop:empty_rr_form}.  Namely, it checks if all states
excluded by $H^*_{i+1}$ are eliminated as unreachable by $U_{j+1}$.
If this is the case, \ti{CompRrForm} returns \ti{nil} reporting that
$C$ is $P^n$-equivalent (line 11). Otherwise, \ti{CompRrForm} switches
to a new time frame by updating set $W$ and formula $\Phi^*$ (lines
12-13).

If none of the formulas $H^*_i$, $i=1,\dots,n$ excludes a reachable
bad state, then from  Proposition~\ref{prop:trivial_case} it follows
that $C$ is $P^n$-equivalent. So \ti{CompRrForm} returns \ti{nil}
(line 14).

%
% Subsection: Searching for a counterexample
%
\subsection{Searching for a counterexample}
\input{e1lim_bad_states.fig} 
%
\mc{} searches for a counterexample by
calling procedure \ti{ElimBadStates}
(Figure~\ref{fig:elim_bad_states}) that, in turn, calls procedure
\ti{PropBack} (Figure~\ref{fig:prop_back}).
The objective of \ti{ElimBadStates} is to show that the bad states
excluded by a range reduction formula $H^*_j$ are unreachable.  This is
done in a \ti{while} loop. First, \ti{ElimBadStates} checks if formula
$\overline{H}_j \wedge U_j \wedge \overline{P}$ is satisfiable (line
2).  Suppose that a satisfying assignment \pnt{p} is found.  Then one
can extract a bad state \pnt{s} from \pnt{p} that is excluded by the
range reduction formula $H^*_j$ and satisfies $U_j$ (line 4).  The
latter means that \pnt{s} has not been proved unreachable yet.  Then
procedure \ti{PropBack} returns a trace from an initial state to
\pnt{s} (if any). If such a trace $E$ exists then \ti{ElimBadStates}
terminates returning $E$. Otherwise, \ti{PropBack} returns a clause
$C$ that is falsified by \pnt{s} thus proving that \pnt{s} is
unreachable. Clause $C$ is added to $U_j$ and a new iteration starts.

\input{prop_back.fig} 
%
The goal of procedure \ti{PropBack} (Figure~\ref{fig:prop_back}) is to
find an initialized trace $E$ leading to the bad state \pnt{s_k}.
Initially $E$ is an empty set (line 1). Trace $E$ is built in the
reverse order. So index $j$ specifying the current time frame is
initialized to $k$ (line 2). The main computation is done in a
\ti{while} loop (lines 3-15).  If $j$ is equal to 0, then the
construction of $E$ is over (line 4). Otherwise, \ti{PropBack} checks
if formula $U_{j-1} \wedge T \wedge \mi{Cnf}(\pnt{s_j}))$ is
satisfiable.  Here $\mi{Cnf}(\pnt{s_j})$ is the set of unit clauses
specifying state \pnt{s_j}. The existence of a satisfying assignment
\pnt{p} means that one can extract an input pair \IP{s_{j-1}}{x_{j-1}}
from \pnt{p} such that
\begin{itemize}
\item \pnt{s_{j-1}} satisfies $U_{j-1}$ and hence is not proved
  unreachable yet
\item system $\xi$ transitions to state \pnt{s_j} under input pair
  \IP{s_{j-1}}{x_{j-1}}.
\end{itemize}

Assume that \pnt{p} does not exist. In this case a resolution proof of
unsatisfiability \ti{Proof} is generated and \ti{PropBack} performs
actions shown in lines 7-12. First, a clause $C$ falsified by
\pnt{s_j} is built. The simplest way to construct $C$ is to negate
$\mi{Cnf}(\pnt{s_j})$. A shorter clause can be generated by excluding
from $C$ the literals that correspond to the unit clauses of
$\mi{Cnf}(\pnt{s_j})$ that were not used in \ti{Proof}. If $j = k$,
then $C$ is falsified by the target state \pnt{s_k} thus proving the
latter unreachable (line 8).  Otherwise, \ti{PropBack} adds $C$ to
$U_j$. After that the input pair \IP{s_j}{x_j} is removed from $E$,
index $j$ is incremented by 1 and a new iteration starts (lines
10-12).

If a satisfying assignment \pnt{p} exists, then an input pair
\IP{s_{j-1}}{x_{j-1}} is extracted from \pnt{p} and added to $E$
(lines 13-14).  The value of $j$ is decremented by 1 and a new
iteration begins (line 15).

%% file: m3c.fig.tex
%
% MC_CRR
%
\setlength{\intextsep}{2pt}
\setlength{\textfloatsep}{10pt}
%\begin{wrapfigure}{L}{2.1in}
\begin{figure}
%\begin{center}
\small
%\normalsize
\vspace{-10pt}
\begin{tabbing}
aaa\=bb\=cc\= dd\= \kill
// $N$ is a comb. circuit specifying transition relation \\
// $I$ is a CNF formula specifying initial states \\
// $P$ is the property to be checked \\
// $n$ is length of the longest counterexample (if any)  \\
// \mc returns a counterexample \\
//  or \ti{nil} if no counterexample exists \\
//\\
$\MC(N,I,P,n)$\{\\
\tb{\scriptsize{1}}\> $T(S,X,S') := \mi{GenCnfForm}(N)$; \\
\tb{\scriptsize{2}}\> $\mi{Trace} := \emptyset$; \\
\tb{\scriptsize{3}}\> $\mi{States} := \emptyset $; \\
\>  - - - - - - - - - - - - -\\
\tb{\scriptsize{4}}\> while (\ti{true}) \{\\
\tb{\scriptsize{5}}\Tt  $\inp{s}{x} := \mi{GenInp}(I)$; \\
\tb{\scriptsize{6}}\Tt  $\pnt{s'} := \mi{Simulate}(P,N,\pnt{s},\pnt{x});$\\
\tb{\scriptsize{7}}\Tt  if $(P(\pnt{s'}) = 0)$ return$(\mi{Trace} \cup \s{(\pnt{s},\pnt{x})})$; \\
\tb{\scriptsize{8}}\Tt  $A := \mi{MinFalsifClause}(\pnt{s},\pnt{x})$; \\
\>  - - - - - - - - - - - - -\\
\tb{\scriptsize{9}}\Tt $E := \mi{ConstrTimeFrame}(T,I,P,A,n)$; \\
\tb{\scriptsize{10}}\Tt if ($E \neq \mi{nil}$) return($\mi{Trace} \cup E$); \\
\tb{\scriptsize{11}}\Tt $\mi{Trace} := \mi{Trace} \cup \s{(\pnt{s},\pnt{x})}$; \\
\tb{\scriptsize{12}}\Tt $\mi{States} := \mi{States} \cup \pnt{s}$; \\
\tb{\scriptsize{13}}\Tt if ($\pnt{s'} \in \mi{States}$) return($\mi{nil}$); \\
\tb{\scriptsize{14}}\Tt $I := \mi{FormUnitClauses}(\pnt{s'})$;\}\} \\
\end{tabbing} 
%\vspace{-20pt}
\caption{Model checking by computing range reduction}
\label{fig:mc_crr}
\end{figure}
%\end{wrapfigure}

%% file: c1onstr_time_frame.fig.tex
%
% ConstrTimeFrame
%
\setlength{\intextsep}{2pt}
\setlength{\textfloatsep}{10pt}
%\begin{wrapfigure}{L}{2.1in}
\begin{figure}
%\begin{center}
\small
%\normalsize
%\vspace{-10pt}
\begin{tabbing}
aaa\=bb\=cc\= dd\= \kill
// $E$ denotes a counterexample \\
// \\
$\mi{ConstrTimeFrame}(T,I,P,A,n)$\{\\
\tb{\scriptsize{1}}\> $G := \emptyset$; \\
\tb{\scriptsize{2}}\> while (\ti{true}) \{\\
\tb{\scriptsize{3}}\Tt   $\inp{s}{x} := \mi{SatAssgn}(G \wedge A)$;\\
\tb{\scriptsize{4}}\Tt   if $(\inp{s}{x} = \mi{nil})$ return($\mi{nil}$); \\
\tb{\scriptsize{5}}\Tt   $C := \mi{GenFalsifClause}(G,A,\pnt{s},\pnt{x})$; // $C \not\rightarrow A$ \\
\tb{\scriptsize{6}}\Tt   $E := \mi{CompRrForm}(T,I,P,C,n)$;\\
\tb{\scriptsize{7}}\Tt   if ($E \neq \mi{nil}$) return($E$);\\
\tb{\scriptsize{8}}\Tt  $G := G \wedge C$; \}\} \\
\end{tabbing} 
%\vspace{-20pt}
\caption{The \ti{ConstrTimeFrame} procedure}
\label{fig:constr_time_frame}
\end{figure}
%\end{wrapfigure}

%% file: c4omp_rr_formulas.fig.tex
%
% CompRrForm
%
\setlength{\intextsep}{2pt}
\setlength{\textfloatsep}{10pt}
%\begin{wrapfigure}{L}{2.1in}
\begin{figure}
%\begin{center}
\small
%\normalsize
%\vspace{-10pt}
\begin{tabbing}
// $T_j = T(S_{j-1},X_{j-1},S'_{j-1})$ \\
// \\
aaa\=bb\=cc\= dd\= \kill
$\mi{CompRrForm}(T,I,P,C,n)$\{\\
\tb{\scriptsize{1}}\> $\rr := \s{H^*_0,\dots,H^*_n}$; \\
\tb{\scriptsize{2}}\> $H^*_0 := \s{C}$; \\
\tb{\scriptsize{3}}\> $\ur := \s{U_1,\dots,U_n}$; \\
\tb{\scriptsize{4}}\> $ W : = S_0 \cup X_0$; \\
\tb{\scriptsize{5}}\> $\Phi^* := I$; \\
\tb{\scriptsize{6}}\> for ($j=0;~j < n;~j\mi{++}$) \{ \\
\tb{\scriptsize{7}}\Tt   $H^*_{j+1} := \mi{SolvePQE}(\prob{W}{H^*_j \wedge T_j \wedge \Phi^*})$;\\
\tb{\scriptsize{8}}\Tt   if (\Nmpl{\overline{H^*_{j+1}}}{P}) \{ \\
\tb{\scriptsize{9}}\ttt     $(E,\ur) := \mi{ElimBadStates}(T,I,P,\rr,\ur,j\!+\!1)$; \\
\tb{\scriptsize{10}}\ttt     if ($E \neq \mi{nil}$) return($E$);\} \\
\tb{\scriptsize{11}}\Tt  if (\Impl{\overline{H^*_{j+1}}}{\overline{U}_{j+1}}) return($\mi{nil}$); \\
\tb{\scriptsize{12}}\Tt  $W := W \cup S_j \cup X_j$; \\
\tb{\scriptsize{13}}\Tt  $\Phi^* := \Phi^* \wedge H^*_j \wedge T_j$; \} \\
\tb{\scriptsize{14}}\> return($\mi{nil}$); \}\\

\end{tabbing} 
%\vspace{-20pt}
\caption{The \ti{CompRrForm} procedure}
\label{fig:comp_rr_form}
\end{figure}
%\end{wrapfigure}

%% file: e1lim_bad_states.fig.tex
%
% ElimBadStates
%
\setlength{\intextsep}{2pt}
\setlength{\textfloatsep}{10pt}
%\begin{wrapfigure}{L}{2.1in}
\begin{figure}
%\begin{center}
\small
%\normalsize
%\vspace{-10pt}
\begin{tabbing}
aaa\=bb\=cc\= dd\= \kill
$\mi{ElimBadStates}(T,I,P,\rr,\ur,j)$\{\\
\tb{\scriptsize{1}}\> while ($\mi{true}$) \{ \\ 
\tb{\scriptsize{2}}\Tt  \pnt{p}=$\mi{SatAssgn}(\overline{H^*_j} \wedge U_j \wedge \overline{P})$;\\
\tb{\scriptsize{3}}\Tt  if ($\pnt{p} = \mi{nil}$) return($\mi{nil},\ur$); \\
\tb{\scriptsize{4}}\Tt  \pnt{s_j} := $\mi{ExtrState}$(\pnt{p}); \\
\tb{\scriptsize{5}}\Tt $(E,C) := \mi{PropBack}(T,I,P,\rr,\ur,\pnt{s_j},j)$;\\
\tb{\scriptsize{6}}\Tt if ($E \neq \mi{nil}$) return($E,\ur$);\\
\tb{\scriptsize{7}}\Tt $U_j := U_j \wedge C$;\}\}   \\
\end{tabbing} 
%\vspace{-20pt}
\caption{The \ti{ElimBadStates} procedure}
\label{fig:elim_bad_states}
\end{figure}
%\end{wrapfigure}

%% file: prop_back.fig.tex
%
% PropBack
%
\setlength{\intextsep}{2pt}
\setlength{\textfloatsep}{10pt}
%\begin{wrapfigure}{L}{2.1in}
\begin{figure}
%\begin{center}
\small
%\normalsize
%\vspace{-10pt}
\begin{tabbing}
// $E$ denotes a counterexample \\
// \\
aaa\=bb\=cc\= dd\= \kill
$\mi{PropBack}(T,I,P,\rr,\ur,\pnt{s_k},k)$\{\\
\tb{\scriptsize{1}}\> $E := \emptyset$; \\
\tb{\scriptsize{2}}\> $j := k$; \\ 
\tb{\scriptsize{3}}\> while ($\mi{true}$) \{\\
\tb{\scriptsize{4}}\Tt   if ($j = 0$) return($E,\mi{nil}$);\\
\tb{\scriptsize{5}}\Tt   $(\pnt{p},\mi{Proof}) := \mi{SatAssgn}(U_{j-1} \wedge T \wedge \mi{Cnf}(\pnt{s_j}))$; \\
\tb{\scriptsize{6}}\Tt   if ($\pnt{p} = \mi{nil}$) \{ \\
\tb{\scriptsize{7}}\ttt     $C := \mi{FormClause}(\mi{Proof},
\mi{Cnfs}(\pnt{s_j}))$; \\
\tb{\scriptsize{8}}\ttt     if ($j = k$) return($\mi{nil},C$); \\
\tb{\scriptsize{9}}\ttt     $U_j := U_j \wedge C$; \\
\tb{\scriptsize{10}}\ttt     $E := E \setminus \s{(\pnt{s_j},\pnt{x_j})}$; \\
\tb{\scriptsize{11}}\ttt    $j := j+1$; \\
\tb{\scriptsize{12}}\ttt    continue; \}\\
\tb{\scriptsize{13}}\Tt    $(\pnt{s_{j-1}},\pnt{x_{j-1}}) = \mi{ExtrInpPair}(\pnt{p})$; \\
\tb{\scriptsize{14}}\Tt   $E := E \cup \s{(\pnt{s_{j-1}},\pnt{x_{j-1}})}$; \\
\tb{\scriptsize{15}}\Tt   $j := j-1$;\}\} \\
\end{tabbing} 
%\vspace{-20pt}
\caption{The \ti{PropBack} procedure}
\label{fig:prop_back}
\end{figure}
%\end{wrapfigure}

%% file: experiments.tex
\section{Experimental results}
\label{sec:experiments}
The key operation of \mc is to compute a range reduction formula by
running a PQE-solver. In this section, we describe experiments meant
to show the viability of using a PQE-solver for computing range
reduction.  We will conduct a more thorough experimental study once
\mc is implemented.  In the experiments, we used the PQE-algorithm
called \pqe~\cite{tech_rep_pqe}.

In addition to showing the viability of using PQE for computing range
reduction, the experiments pursued three other goals. The first goal was
to show that PQE can be much more efficient than QE. The second
goal was to demonstrate that reducing the noise generated by a PQE-solver
can significantly improve its performance. This third goal was to
show that the set of states excluded by range reduction formulas
is drastically smaller than the set of reachable states.

%
% Subsection:
%
\subsection{Using PQE-solver to compute range reduction}
In this subsection, we describe experiments with computing range
reduction by a PQE-solver. In these experiments, we used 758
benchmarks of HWMCC-10 competition.  Let $N(X,S,S')$ be the circuit
representing a transition relation and $T$ be a CNF formula specifying
$N$. Recall that $X$ denotes the input combinational variables and
$S,S'$ denote the present and next state variables.  (For the sake of
simplicity we do not mention the internal combinational variables of
$N$.)  So $S \cup X$ and $S'$ specify the input and output variables
of $N$ respectively.

In experiments, we computed the range reduction of $N$ caused by
excluding the input pairs \IP{s}{x} falsifying a clause $C(S,X)$.  We
used two methods of computing range reduction. The first method was just
to run \PQE on formula \prob{W}{C \wedge T} where $W = S \cup X$.  The
second method first optimized clause $C$ to reduce the amount of noise
generated by \pqe. This optimization was performed by a technique 
called \ti{clause expansion}.  The idea of clause expansion is to
replace $C$ with a clause $C \vee lit(v)$ if the clause $C \vee
\overline{\mi{lit}}(v)$ is implied by $T$. Here $\mi{lit}(v)$ is a
literal of variable $v$. It is not hard to show that, in this case,
taking out clause $C$ from \prob{W}{C \wedge T} is equivalent to
taking out $C'$ from \prob{W}{C' \wedge T} where $C' = C \vee
\mi{lit}(v)$.  The objective of replacing $C$ with $C \vee \mi{lit}(v)$
is to reduce noise generation by removing the part of $C$ that is
implied by $T$. Note that clause $C \vee \mi{lit}(v)$ can be further
expanded.  So, in the second method, \PQE was applied to formula
\prob{W}{C^* \wedge T} where $C^*$ was obtained from $C$ by adding
literals. We used a very efficient procedure of clause expansion. We
omit the details of this procedure.

\input{r1ng_cactus.fig}

For every transition relation out of 758, we generated a random clause
$C$ of length $|S|$ consisting of literals of $S$.  (Note that the
total set of input variables of circuit $N$ specifying a transition
relation is $S \cup X$. So clause $C$ excluded $2^k$ input assignments
where $k=|X|$.  In many transition relations $k$ was greater than
100.)  Then we tried to check range reduction by the two methods
above. We ran many experiments generating different clauses for the
same transition relation.  Here are the results of a typical run
consisting of 758 problems where for every transition relation one
clause of $|S|$ literals was generated randomly.  When using the first
method, only 452 out of 758 problems were finished within the 60s time
limit. The second method succeeded in 733 out of 758 problems. Most of
them were finished within a second.

Let $H(S')$ and $H^*(S')$ denote a noise-free and noisy solution to
the PQE problem \prob{W}{C \wedge T} respectively.  That is $H^*
\wedge \prob{W}{T} \equiv H \wedge \prob{W}{T} \equiv \prob{W}{C
  \wedge T}$ and \Impl{H^*}{H}. Our PQE-solver \PQE is noisy. So it
generates formula $H^*$ rather than $H$. In 643 out of 733 problems
solved by the second method, $H^*$ was empty meaning that no range
reduction occurred. In this case $H^* \equiv H$. In the remaining 90
solved problems, in 2 cases, $T \rightarrow H^*$ held. So here, no
range reduction occurred either and \PQE just generated noise. In 88
cases, $T \rightarrow H^*$ did not hold, indicating that the range of
$T$ reduced.

The performance of these two methods of computing range reduction in
the run we described above is shown in
Figure~\ref{fig:range_reduction}. This figure also provides data on
computing the \ti{full range} of transition relations for the 758
HWMCC-10 benchmarks. As we discussed in
Subsection~\ref{subsec:qe_pqe_range}, finding the full range of a
combinational circuit reduces to QE. So comparing methods for
computing full range and range reduction is a way to compare QE and
PQE.  In the experiments, we computed full range by the QE-algorithm
called \CDI ~\cite{fmcad13}.  With the time limit of 60s, \Cdi
finished only for 62 transition relations.

Figure~\ref{fig:range_reduction} shows the number of problems finished
in a given amount of time.  This data indicates that PQE can be
dramatically more efficient that QE.  One more conclusion that can be
drawn from Figure~\ref{fig:range_reduction} is that reducing noise
generation, e.g. by clause expansion, can have a drastic effect on the
performance of a PQE algorithm.

%
% Subsection
%
\subsection{Comparing set of states describing range
 reduction  with that of reachable states} 
Let $N$ be a circuit specifying a transition relation.  Let $H$ be a
noise-free formula describing the range reduction of  $N$ caused
by constraining inputs with clause $C$. Assume that $C$ depends only
on state variables. In terms of Subsection~\ref{subsec:main_idea}, the
set of states falsifying $H$ can be represented as $A \setminus
B$. Here $A$ consists of the states that are reachable from the states
falsified by $C$ in one transition.  The set $B$ is a subset of
$A$. It consists of the states of $A$ that are also reachable in one
transition from states \ti{satisfying} clause $C$. So $A \setminus B$
consists of states that are reachable \ti{only} from states falsifying
clause $C$.

In Subsection~\ref{subsec:main_idea}, we conjectured that the set $A
\setminus B$ can be dramatically smaller than sets $A$ and $B$. In
this subsection, we check this conjecture by comparing the size of the
set $A \setminus B$ and $A$ experimentally. Given a clause $C$,
computing the set $A \setminus B$ comes down to finding a range
reduction formula $H$. Formula $H$ is obtained by taking $C$ out of
the scope of quantifiers in \prob{W}{C \wedge T}.  Since \PQE produces
a noisy solution $H^*$, we considered only the cases where $H^*$ was
empty and hence $H^* \equiv H$.  Building set $A$ comes down to
finding a quantifier-free formula $G(S')$ that is logically equivalent
to \prob{W}{\overline{C} \wedge T}. (So finding $G$ reduces to the QE
problem.) The set $A$ is specified by the complete assignments to $S'$
satisfying $G$.  To estimate the size of $A$ we generated a limited
number of cubes containing satisfying assignments of $G$. The size of
the largest cube was used as a lower bound on the size of set $A$.

To compute formula $G$ we used our QE solver called \Cdi mentioned
above. In this experiment, we used the same 758 transition relations
of the HWMCC-10 benchmark set.  To make generation of formula $G$ less
trivial we generated a clause of of $0.7 \times |S|$ literals (as
opposed to clauses of $|S|$ literals generated in the previous
experiment). 

Here are the results of a typical experiment consisting of solving 758
PQE and 758 QE problems. Every PQE problem is to take $C$ out of the
scope of quantifiers in \prob{W}{C \wedge T}. The corresponding QE
problem is to eliminate quantifiers in formula \prob{W}{\overline{C}
  \wedge T}.  With the time limit of 60s, \PQE solved 561 PQE problems
while \Cdi solved only 377 QE problems. In 490 PQE problems, an empty
formulas $H^*$ were generated i.e. set $A \setminus B$ was empty.  In
347 out of these 490 cases, the corresponding QE problem was solved by
\CDI. In 211 out of 347 cases, the size of set $A$ was larger than
$2^{30}$ states. In 92 out of 347 cases, the size of set $A$ was larger
than $2^{100}$ states.

\input{table1}

Some concrete results are shown in Table~\ref{tbl:concr_exmps}.  The
first column gives benchmark names. The next three columns specify the
size of the circuit $N$ specifying a transition relation, the column
\ti{\#X-inputs} giving the number of combinational inputs of $N$. The
next two columns give the time taken to solve the corresponding PQE
and QE problems in seconds. The last column provides the lower bound
on the size of set $A$.  For all the examples listed in
Table~\ref{tbl:concr_exmps} the set $A \setminus B$ was empty. The
results of Table~\ref{tbl:concr_exmps} show that the size of the set
$A \setminus B$ can be very small even when sets $A$ and $B$ are very
large.

%% file: r1ng_cactus.fig.tex
\setlength{\intextsep}{4pt}
\begin{figure}
%\begin{wrapfigure}{l}{2in}
 \begin{center}
% es\includegraphics[width=2in]{color_cactus.png}
\includegraphics[width=3in,height=2in]{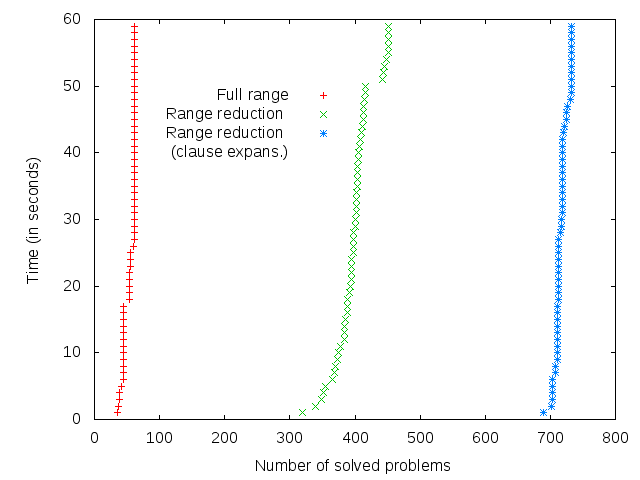}
  \end{center}
\caption{Computing full range and range reduction}
\vspace{5pt}
\label{fig:range_reduction}
%\end{wrapfigure}
\end{figure}

%% file: table1.tex
%
%
%\begin{wraptable}{l}{2.2in}
\vspace{5pt}
\begin{table}[htb]
\small
\caption{\ti{Estimating the size of set $A$ for  some concrete examples}}
\setlength\extrarowheight{2pt}
\vspace{-10pt}
\scriptsize
\begin{center}
%\begin{tabular}{|l|l|l|l|l|l|} \hline
\begin{tabular}{|p{42pt}|p{20pt}|l|p{14pt}|p{12pt}|p{18pt}|p{28pt}|} \hline
 benchmark & \#X-in- & \#lat- & \#gates & PQE & QE   & size   \\ 
           &   puts &  ches       & & (s.) & (s.)   & of $A$   \\  \hline

brpp1neg      & 86  & 138 & 1,244 & 0.01 & 0.01  & $> 2^{84}$   \\ \hline

eijks1423 & 17 & 157  & 1,101  & 0.01  & 18  & $ > 2^{31}$    \\ \hline

bc57sensorsp0 & 97  & 167 &1,691  & 0.01  & 0.4  &$>2^{105}$    \\ \hline

irstdme6 & 220 & 245 & 1,713 & 0.01  & 0.2  & $ > 2^{101}$   \\ \hline

csmacdp0neg  &  146 & 265  & 5,247 & 0.01  & 13 & $>2^{169}$   \\ \hline

139452p24  & 225 & 314 & 5,867 & 0.03 & 0.1  & $>2^{240}$   \\ \hline

pj2013 & 1,305 & 1,271 & 35,630  & 0.2 & 0.1  & $>2^{1231}$    \\ \hline

neclaftp1001 & 32 & 7,880 & 63,383  & 0.3  & 3.8  & $> 2^{2252}$   \\ \hline

\end{tabular}                
\label{tbl:concr_exmps}
\end{center}
\end{table}

%% file: c5onclusions.tex
\section{Conclusions}
\label{sec:conclusions}
We presented a new method of model checking based on the idea of
Computing Range Reduction (CRR). The CRR method repetitively applies
an operation that reduces the set of possible behaviors.  Given a
number $n$ and a property $P$, the CRR method finds a counterexample
of length at most $n$ or proves that such a counterexample does not
exit.  A key feature of the CRR method is that it has a natural way to
do bug hunting in a depth-first manner.

The results of this paper lead to the following conclusions.  First,
computing an \ti{under-approximation} of available behaviors is
complementary to current methods that over-approximate such behaviors.
Computing an under-approximation seems to be a reasonable idea in case
of bug-hunting, because one of the main concerns here is to reduce the
search space.  Second, a successful bug hunting tool should be able to
efficiently perform depth-first search. Third, the scalability issues
of model checkers based on Quantifier Elimination (QE) are caused by
the fact that QE is an inherently hard problem and should be avoided.
The key operation of the CRR method is based on \ti{partial} QE.  In
many cases, the partial QE problem can be solved dramatically more
efficiently than its QE counterpart.  Fourth, the partial QE problem
cannot be efficiently solved by a typical SAT-solver based on the
notion of logical inconsistency rather than unobservability. So,
development of non-SAT methods of model checking is very important.

%% file: p3roofs.tex
%\clearpage
\appendix
\setcounter{proposition}{0}
%\vspace{10pt}
%\noindent\tb{\large{Appendix}} \\
%\vspace{10pt}

The appendix contains   proofs of the propositions  listed in the paper.
We also give proofs of lemmas  used in the proofs of propositions.

\section*{Propositions of Section~\ref{sec:pqe}: Partial Quantifier Elimination}
%
% Proposition: QE as range computation 
%
\begin{proposition}
%\label{prop:full_range}
Let $R(Z)$ be a CNF formula such that $R \equiv \prob{W}{G}$ where
$W = X \cup Y$. (That is $R$ is a solution to the QE problem.)
Then the assignments satisfying $R(z)$ specify the range of $M$.
\end{proposition}
% Proof
\begin{proof}
Let us show that $R$ indeed specifies the range of $M$. Let \pnt{z} be
a complete assignment to $Z$ that is in the range of $M$. Then there
is an assignment $(\pnt{x},\pnt{y},\pnt{z})$ satisfying $G$ and hence
\prob{W}{G} evaluates to 1 when variables $Z$ are assigned as in 
\pnt{z}. Hence $R(\pnt{z})$ has to be equal to 1.  Now assume that
\pnt{z} is not in the range of $M$. Then no assignment
$(\pnt{x},\pnt{y},\pnt{z})$ satisfies $G$. So \prob{W}{G} evaluates to
0 for assignment \pnt{z}. Then $R(\pnt{z})$ has to be equal to 0.
\end{proof}

%
% Proposition: PQE as range reduction computation
%
\begin{proposition}
%\label{prop:range_red}
Let $C(X)$ be a clause depending only on input variables of $M$.  Let
$H(Z)$ be a CNF formula such that $H \wedge \prob{W}{G} \equiv
\prob{W}{C \wedge G}$ where $W = X \cup Y$. (That is $H$ is a solution
to the PQE problem.)  Let $H$ and $H^*$  be  a noise-free and noisy 
solution respectively.  Then
\begin{enumerate}
\item The assignments \ti{falsifying} $H$ specify the range reduction
  in $M$ caused by excluding inputs falsifying $C$. That is $H(\pnt{z})=0$
  iff
\begin{itemize}
   \item there is an input \pnt{x} for which circuit $M$ produces
   output \pnt{z} 
   \item all inputs for which $M$ produces
   output \pnt{z} falsify $C$
\end{itemize}
\item \Impl{H^*}{H}
\end{enumerate}
\end{proposition}
% Proof
\begin{proof}
\tb{First condition.} Let us prove that $H$ indeed specifies the range
reduction of $M$. Let \pnt{z} be a complete assignment to $Z$ that is
in the range of $M$.  Assume that \pnt{z} remains in the range of $M$
even if the inputs falsifying clause $C$ are excluded. Then there is
an assignment $(\pnt{x},\pnt{y},\pnt{z})$ satisfying $C \wedge G$ and
hence \prob{W}{C \wedge G} evaluates to 1 when variables of $Z$ are
assigned as in \pnt{z}. So, $H(\pnt{z})$ has to be equal to 1.

Now assume that \pnt{z} is in the range of $M$ but it is not in the
range of $M$ if the inputs falsifying clause $C$ are excluded. Then no
assignment $(\pnt{x},\pnt{y},\pnt{z})$ satisfies $C \wedge G$ and
hence \prob{W}{C \wedge G} evaluates to 0 when variables of $Z$ are
assigned as in \pnt{z}. On the other hand, since \pnt{z} is in the
range of $M$, there is an assignment $(\pnt{x},\pnt{y},\pnt{z})$
satisfying $G$.  So formula \prob{W}{G} is equal to 1 when variables
of $Z$ are assigned as in \pnt{z}. Since $H \wedge \prob{W}{G}$ is
equal to 0 when variables of $Z$ are assigned as in \pnt{z}, then
$H(\pnt{z})$ has to be equal to 0.

 Now assume that \pnt{z} is not in the range of $M$. Then no
 assignment $(\pnt{x},\pnt{y},\pnt{z})$ satisfies $G$. So
 both \prob{W}{G} and \prob{W}{C \wedge G} evaluate to 0. This means
 that the value of $H(\pnt{z})$ is, in general, not defined. However,
 since we require $H$ to be a noise-free solution, $H(\pnt{z})$ has to
 be equal to 1.

\tb{Second condition.} As we showed above, any solution to the PQE
problem is defined uniquely for a complete assignment \pnt{z} to $Z$
that is in the range of $M$. So in this case,
$H(\pnt{z})=H^*(\pnt{z})$.  If \pnt{z} is not in the range of $M$, by
definition of a noise-free solution, $H(\pnt{z})=1$.  So
$H(\pnt{z})=0$ implies $H^*(\pnt{z})=0$  and hence \Impl{H^*}{H}.
\end{proof}

\section*{Propositions of Section~\ref{sec:isol_publ_traces}: Isolated And Public Traces}
%
% Lemma: about trace specified by range reduction formulas
%
\begin{lemma}
\label{lemma:trace}
Let $H_1,\dots,H_m$ be range reduction formulas computed with respect
to clause $C$. Let $E$ be an initialized trace \tr{s}{x}{0}{m}
such that 
\begin{itemize}
\item \IP{s_0}{x_0} falsifies $C$ i.e. $E$ is excluded by $C$
\item \IP{s_i}{x_i} falsifies $H_i$, $i=1,\dots,m$.
\end{itemize} 
  Let $E'$ be an initialized
trace \tr{s'}{x'}{0}{m} that is allowed by clause $C$.  Then
$H_i(\pnt{s'_i})=1$, $i=1,\dots,m$ and hence $\pnt{s_i} \neq
\pnt{s'_i}, i=1,\dots,m$.
\end{lemma}
\begin{proof}
Since \pnt{s'_i} is in $E'$, it is reachable by a trace allowed by
$C$. From Definition~\ref{def:rr_formulas} it follows that
$H_i(\pnt{s'_i})=1$.
\end{proof}

%
% Proposition: about isolated traces
%
\begin{proposition}
%\label{prop:isol_traces}
Let $C(S,X)$ be a clause such that \Impl{\overline{C}}{I}. Let
$H_1,\dots,H_m$ be range reduction formulas computed with respect to
clause $C$. Let $E$ denote an initialized trace \tr{s}{x}{0}{m} such
that
\begin{itemize}
\item \IP{s_0}{x_0} falsifies $C$ i.e. $E$ is excluded by $C$
\item \IP{s_i}{x_i} falsifies $H_i$, $i=1,\dots,m$.
\end{itemize} 
 Then $E$ is isolated with respect to $C$.
\end{proposition}

\begin{proof}
Assume that $E$ is not isolated. 
Then there is an initialized  trace $E' =  \tr{s'}{x'}{0}{m}$ such that
\begin{itemize}
\item $E'$ is allowed by $C$
\item \pnt{s_i}=\pnt{s'_i}, for some $i$ such that $1 \leq i \leq m$.
\end{itemize}
The existence of such a trace contradicts Lemma~\ref{lemma:trace}.
\end{proof}

%
% Proposition: P-equivalence of a clause
%

\begin{proposition}
%\label{prop:p_equiv_clause}
Let $\xi$ be a system with property $P$. Let $C(S,X)$ be a clause such
that \Impl{\overline{C}}{I}. Assume that $C$ does not exclude any
counterexample of length at most $n$ isolated with respect to $C$.
Then $C$ is a $P^n$-equivalent clause.
\end{proposition}
\begin{proof}
Assume the contrary i.e. every counterexample of length at most $n$ is
excluded by $C$ and so $C$ is not $P^n$-equivalent. Let
$E$=\tr{s}{x}{0}{m} be a counterexample of length $m \leq n$ excluded
by $C$.  By assumption, $E$ is not isolated with respect to $C$. Then
there is an initialized trace $E'$ equal to \tr{s'}{x'}{0}{k}, $k \leq
m$ such that
\begin{itemize}
\item $E'$ is allowed by $C$
\item \pnt{s'_k} = \pnt{s_k}.
\end{itemize}
Let $E''$ be a sequence of input pairs
\tr{s'}{x'}{0}{k},\tr{s}{x}{k+1}{m}. Since $E''$ is obtained by
stitching together two traces and \pnt{s'_k} = \pnt{s_k}, $E''$ is a
trace.  Since \pnt{s'_0} is an initial state, $E''$ is an initialized
trace. Since $\xi$ transitions to a bad state under
input \IP{s_m}{x_m} $E''$ is a counterexample.  Since \IP{s'_0}{x'_0}
satisfies $C$, $E''$ is allowed by $C$.  So $C$ does not exclude all
counterexamples of length at most $n$ and we have a contradiction.
\end{proof}

\section*{Propositions of Section~\ref{sec:bug_hunting}: Bug Hunting By CRR}

%
%  Proposition: CRR by a noise-free PQE solver
%
\begin{proposition}
%\label{prop:noise_free_pqe}
Let $\xi$ be a state transition system with property $P$. Let $C(S,X)$
be a non-empty clause such that \Impl{\overline{C}}{I}.  Let $H_0$
denote formula equal to $C$. Let formulas $H_1,\dots,H_n$ be obtained
recursively as follows.  Let $\Phi_0$ denote formula equal to $I$.
Let $\Phi_i,~~0 < i \leq n$ denote formula $I \wedge H_0 \wedge
T_0 \wedge \dots \wedge H_{i-1} \wedge T_{i-1}$. Here $T_j =
T(S_j,X_j,S_{j+1})$ where $S_j$ and $X_j$ are state and input
variables of $j$-th time frame.  Formula $H_{i+1}$ is obtained by
taking $H_i$ out of the scope of quantifiers in
formula \prob{W}{H_i \wedge T_i \wedge \Phi_i} where $W = S_0 \cup
X_0 \cup \dots \cup S_i \cup X_i$.  That is
$H_{i+1} \wedge \prob{W}{T_i \wedge \Phi_i} \equiv \prob{W}{H_i \wedge
T_i \wedge \Phi_i}$.  Then formulas $H_1,\dots,H_n$ are range
reduction formulas.
\end{proposition}
\begin{proof}
Let us prove this proposition by induction on $i$. This proposition is
vacuously true for $i=0$. Assume that it holds for $i=0,\dots,n$. Let
us show that then this proposition holds for $n+1$.  That is one needs
to show that $H_{n+1}$ is a range reduction formula and hence
$H_{n+1}(\pnt{s})=0$ iff \pnt{s} is reachable in $n+1$ transitions only
by traces excluded by $C$.  Assume that $H_{n+1}$ is not a range
reduction formula.  Then one needs to consider the two cases below.

A) $H_{n+1}(\pnt{s})=0$ and \pnt{s} is not reachable by any
initialized trace of length $n+1$. This means that \pnt{s} cannot be
extended to satisfy formula $I \wedge T_0 \dots \wedge T_n$.
Hence \pnt{s} cannot be extended to satisfy formula $\Phi_n \wedge
T_n \wedge H_n$. Then the clause of maximal length falsified
by \pnt{s} is implied by $\Phi_n$. This means that $H_{n+1}$ is a
``noisy'' solution of the PQE problem and hence cannot be obtained by
a noise-free PQE solver. So we have a contradiction.

B) The set of initialized traces of length $n+1$ reaching
state \pnt{s} is not empty but at least one trace of this set is
allowed by $C$. Let $E$ be such a trace. The fact that $E$
reaches \pnt{s} means that $E$ satisfies formula $I \wedge
T_0 \dots \wedge T_n$. Since $E$ is a trace allowed by $C$ it also
satisfies $C$. Moreover, $E$ has to satisfy all the formulas $H_i$,
$i=1,\dots,n$. Indeed, if $E$ falsifies $H_i$ then there is a
initialized trace of length $i$ that is allowed by $C$ and that
reaches a state excluded by $H_i$. This means that $H_i$ is not a
range reduction formula.  So $E$ satisfies $H_i$,$i=1,\dots,n$ and
hence formula $H_n \wedge T_n \wedge \Phi_n$ is satisfied by $E$. This
means that formula $H_{n+1}$ is not implied by $H_n \wedge
T_n \wedge \Phi_n$. Hence $H_{n+1} \wedge \prob{W}{T_n \wedge \Phi_n}$
is not equivalent to
\prob{W}{H_n \wedge T_n \wedge \Phi_n} and we have a contradiction.
\end{proof}
%
% Lemma about implication
%
\begin{lemma}
\label{lemma:impl}
Let $F',F'',H',H'',G$ be  CNF formulas such that
\begin{itemize}
\item $H' \wedge \prob{X}{G} \equiv \prob{X}{F' \wedge G}$
\item $H'' \wedge \prob{X}{G} \equiv \prob{X}{F'' \wedge G}$
\item \Impl{F'}{F''}
\end{itemize}
Let $H'$,$H''$ be obtained by a noise-free PQE solver.
Then \Impl{H'}{H''} holds.
\end{lemma}
\begin{proof}
Let $Y$ denote the set of free variables. Assume the contrary
i.e. \Nmpl{H'}{H''}.  Then there is a complete assignment \pnt{y} to
$Y$ such that $H'(\pnt{y})=1$ and $H''(\pnt{y})=0$. The latter means
that 
\begin{enumerate}
\item \prob{X}{F'' \wedge G}=0 in subspace \pnt{y} and so every
assignment (\pnt{x},\pnt{y}) falsifies $F'' \wedge G$
\item Since $H''$ is obtained by a noise-free PQE solver,
\Nmpl{G}{C} where $C$ is the longest clause falsified by \pnt{y}.
So there is an assignment (\pnt{x},\pnt{y}) satisfying $G$.
\end{enumerate}

The fact that every assignment (\pnt{x},\pnt{y}) falsifies $F'' \wedge
G$ and that \Impl{F'}{F''} entails that every assignment
(\pnt{x},\pnt{y}) falsifies $F' \wedge G$ as
well. So \prob{X}{F' \wedge G}=0 in subspace \pnt{y}. This means that
$H' \wedge \prob{X}{G} = 0$ in subspace \pnt{y} as well. The fact that
there is an assignment (\pnt{x},\pnt{y}) satisfying $G$ and $H'$
depends only on variables of $Y$ implies that $H'(\pnt{y})=0$.
So we have a contradiction.
\end{proof}
%
% Proposition: CRR by a noisy PQE solver
%
\begin{proposition}
%\label{prop:noisy_pqe}
Let $H^*_i, i=0,\dots,n$ be formulas obtained as described in
Proposition~\ref{prop:noise_free_pqe} with only one exception. A \ti{noisy}
PQE-solver is used to obtain $H^*_{i+1}$ by taking $H^*_i$ out of the
scope of quantifiers in \prob{W}{H^*_i \wedge T_i \wedge \Phi^*_i}.
Here $\Phi^*_0 = I$, $H^*_0=C$ and $\Phi^*_i = I \wedge H^*_0 \wedge
T_0 \wedge \dots \wedge H^*_{i-1} \wedge T_{i-1}$ for $i < 0 \leq n$.
Then \Impl{H^*_i}{H_i} holds where $H_i,i=1,\dots,n$ are range
reduction formulas.
\end{proposition}
\begin{proof}
We prove this proposition by induction on $i$. 
\Impl{H^*_0}{H_0} holds because $H^*_0 = H_0 = I$.
Now we prove that \Impl{H^*_i}{H_i}, $i \geq 0$ entails that
\Impl{H^*_{i+1}}{H_{i+1}}. Denote by $Q_{i+1}$ a noise-free formula
obtained by taking $H^*_i$ out of the scope of quantifiers
in \prob{W}{H^*_i \wedge T_i \wedge \Phi_i}. From
Lemma~\ref{lemma:impl} it follows that \Impl{Q_{i+1}}{H_{i+1}}. On the
other hand, from Proposition~\ref{prop:range_red} it follows that
\Impl{H^*_{i+1}}{Q_{i+1}}. Hence \Impl{H^*_{i+1}}{H_{i+1}}.
\end{proof}

\section*{Propositions of Section~\ref{sec:p_equiv_clauses}: Generation Of $P$-Equivalent Clauses}
\begin{proposition}
%\label{prop:trivial_case}
Let $\xi$ be a system with property $P$. Let $C(S,X)$ be a clause such
that \Impl{\overline{C}}{I}. Let $H^*_1,\dots,H^*_n$ be approximate
range reduction formulas computed with respect to clause $C$ by a
noisy PQE solver.  Suppose that no formula $H^*_i$,~$i=1,\dots,n$
excludes a reachable bad state \pnt{s}. Then clause $C$ is
$P^n$-equivalent.
\end{proposition}
\begin{proof}
From Proposition~\ref{prop:noisy_pqe} it follows
that \Impl{H^*_i}{H_i} where $H_i$ is a range reduction formula.  So
that fact that $H^*_i$ does not exclude a bad reachable state implies
that $H_i$ does not exclude a bad state. This means that clause $C$
does not exclude an isolate counterexample of length at most $n$. Then
Proposition~\ref{prop:p_equiv_clause} entails that $C$ is
$P^n$-equivalent.
\end{proof}

%
% P^n-equivalent clause when H_i is empty
%
\begin{proposition}
%\label{prop:empty_rr_form}
Let $\xi$ be a system with property $P$. Let $C(S,X)$ be a clause such
that \Impl{\overline{C}}{I}. Let $H^*_1,\dots,H^*_i$ be approximate
range reduction formulas computed with respect to clause $C$ by a
noisy PQE solver. Suppose that every bad state excluded by $H^*_j$, $1
\leq j < i$ is unreachable. Suppose that every state (bad or good)
excluded by $H^*_i$ is unreachable. Then clause $C$ is
$P^n$-equivalent for any $n > 0$.
\end{proposition}
\begin{proof}
From Proposition~\ref{prop:noisy_pqe} it follows
that \Impl{H^*_j}{H_j} where $H_j$ is a range reduction formula.  So
that fact that $H^*_j$, $1 \leq j < i$ does not exclude a bad
reachable state implies that $H_j$ does not exclude a bad state.  The
fact that every state excluded by $H^*_i$ is unreachable means that
$H_i$ is empty i.e. $H_i \equiv 1$. Formula $H_{i+1}$ is obtained by
taking $H_i$ out of the scope of quantifiers in
formula \prob{W}{H_i \wedge T_i \wedge \Phi_i}. This means that
$H_{i+1} = \equiv 1$ and hence $H_{i+1}$ does not exclude any bad states
either. So no formula $H_i$, $i > 0$ excludes a bad state. Hence clause
$C$ is $P^n$ equivalent for any $n > 0$.
\end{proof}